\documentclass[12pt,a4paper, twoside,reqno]{amsart}
\usepackage[margin=3cm]{geometry}

\usepackage{tikz}
\usepackage{amsmath}
\usepackage{amssymb,amsfonts,mathrsfs}
\usepackage[colorlinks=true,linkcolor=blue,citecolor=blue]{hyperref}

\usepackage{mllocal} 

\numberwithin{equation}{section}

\newcommand{\pa}{\partial}
\newcommand{\bu}{\bullet}

\newcommand{\CC}{\mathbb C}

\begin{document}


\title[Gluing Formula for Determinants]
{Combinatorial Quantum Field Theory and Gluing Formula for Determinants}

\author{Nicolai Reshetikhin}
\address{N.R.: Department of Mathematics, University of California,
Berkeley,
CA 94720, USA \& KdV Institute for Mathematics, University of Amsterdam,
Science Park 904, 1098 XH Amsterdam, The Netherlands \& ITMO University,
Kronverkskii ave. 49, Saint Petersburg 197101, Russia.}

\email{reshetik@math.berkeley.edu}
\urladdr{http://math.berkeley.edu/~reshetik}

\author{Boris Vertman}
\address{B.V.: Mathematisches Institut,
Universit\"at M\"unster,
Germany}
\email{vertman@uni-muenster.de}

\subjclass[2000]{58J52; 81T27}
\date{This document was compiled on: \today}

\begin{abstract}
{We define the combinatorial Dirichlet-to-Neumann operator
and establish a gluing formula for determinants of discrete Laplacians
using a combinatorial Gaussian quantum field theory. In case of a diagonal
inner product on cochains we provide an explicit local expression for the discrete
Dirichlet-to-Neumann operator. We relate the gluing formula to the corresponding Mayer-Vietoris formula
by Burghelea, Friedlander and Kappeler for zeta-determinants of analytic Laplacians,
using the approximation theory of Dodziuk. Our argument motivates
existence of gluing formulas as a consequence of a gluing principle on the discrete level.}
\end{abstract}

\maketitle

\section{Introduction}\label{intro}

\subsection{Gluing formula for zeta-determinants}
Investigation of the cut and paste behaviour of
zeta-regularized determinants has been initiated by
Forman \cite{For:FDA} and Burghelea-Friedlander-Kappeler in
\cite{BFK:MVT}. Their studies have triggered further
analysis by various authors including Park-Wojciechowski
\cite{ParWoj:ADO} and Lee \cite{Lee:BFK}.

The approach of \cite{For:FDA} and \cite{BFK:MVT} to the gluing behavior
of zeta-regularized determinants is purely analytic.
In the present paper we use quantum
field theoretic arguments to establish a gluing formula
for determinants of combinatorial Laplacians and
relate it to the gluing formulas for zeta-determinants of analytic Laplacians
using the approximation theory of Dodziuk \cite{Dod:FDA},
and its extension to manifolds with boundary by M\"uller
\cite{Mue:ATA}.

Our arguments naturally motivate
existence of gluing formulas for the non-local spectral invariants
such as zeta-regularized determinants as a consequence of the corresponding
gluing principle in the discrete case. The underlying analysis is firmly embedded
in the formalism of a discrete Gaussian quantum field theory according to the 
axioms of Atiyah-Segal. In fact we provide an explicit construction of 
such a theory here.

Two main results  are gluing formula for combinatorial Laplacians, and the
continuum limit for the ratio of determinants of Laplace operators.
An important part of the paper is the formulation of Gaussian quantum field theory
for Whitney product on cochains.

\subsection{Gluing formula for combinatorial determinants} \label{int-2}
Consider an $n$-dimensional simplicial complex $K$,
possibly with an $(n-1)$ dimensional boundary $L$. Assume that the boundary subcomplex
$L \subset K$ has three connected components $L_1, L_2$ and $L_3$. We write
$$L=L_1 \sqcup L_2 \sqcup L_3.$$

Assume that $K$ is equipped with a Riemannian structure (see \S \ref{comb}) and
that $L_1$ and $L_2$ are isomorphic via an isometry $f$ of simplicial complexes with Riemannian
structure. Gluing the boundary components $L_1$ and $L_2$ via the isometric identification $f$
defines a new Riemannian simplicial complex $K_f$ with a single boundary
subcomplex $L_3$.

Consider a combinatorial Hermitian vector bundle $E$ over $K_f$ with
a flat connection $\A$, which gives rise to combinatorial coboundary
operators on the cochains. $E$ pulls back to a vector bundle over $K$
under the map of simplicial complexes $K\to K_f$ introduced above.
This defines a scalar product on cochains with values in the corresponding 
vector bundle and gives rise to the combinatorial
Laplacians $\Delta_{K_f}$ and $\Delta_{K}$ on cochains of degree zero with Dirichlet boundary conditions at
the corresponding boundary complexes.

Let $Q_K$ be the linear operator connecting the scalar products on $C^0(K,E)$,
corresponding to the given Riemannian structure on $K$ and the diagonal (in the simplex basis)
Euclidean inner product. More explicitly,
let $\langle \cdot , \cdot \rangle_0$ denote the Euclidean inner product on cochains. Then, for
any $\phi, \psi \in C^0(K,E)$ we have $\langle \phi, \psi \rangle_K = \langle Q_K \phi, \psi \rangle_0$.
We write $\Delta^{loc}_{K}:= Q_{K\backslash L} \circ \Delta_{K}$ and
$\Delta^{loc}_{K_f}:= Q_{K_f\backslash L_3} \circ \Delta_{K_f}$ for the \emph{local} Laplacians with respect to
the Euclidean inner product on cochains with Dirichlet boundary conditions.
In \S \ref{QFT} we prove the following gluing formula for determinants of combinatorial
local Laplacians.

\begin{thm}\label{main1}
Assume the Riemannian structure on $K$ is local.
Then the determinants of combinatorial Laplacians satisfy the following identity
\begin{align*}
  \frac{\det' (\Delta^{loc}_{K_f})}
{\det' (\Delta^{loc}_{K})}
= \det\nolimits' \mathscr{R}^{loc}_c(K_f,L_2),
\end{align*}
where $\det'$ is the product of non-zero eigenvalues and $\mathscr{R}^{loc}_c(K,L_2)$ is the composition
of $Q_{L_2}$ with the combinatorial analog of the
Dirichlet-to-Neumann map defined in (\ref{DN-gluing}).
\end{thm}

The operators $Q_K$ satisfy their own intricate gluing law which we make explicit at the end of \S \ref{QFT}.

\subsection{Relating both gluing formulae in the discretization limit}\label{relating-gluing}
In \S \ref{approximation} we establish the relationship between this gluing formula and
the gluing formula for the analytic zeta-regularized
determinants.

Let $(M,g)$ be a compact Riemannian manifold
with boundary $\partial M$ that consists of three disjoint boundary components
$N_1,N_2$ and $N_3$. Consider a flat Hermitian vector bundle $(E,h, \A)$, i.e. Hermitian bundle with a Hermitian flat connection $\A$. Flatness implies product structure of $h$ and $\A$ over a collar neighborhood of $\partial M$. We denote by $\Delta_{M}$ the Laplace Beltrami operator acting of functions on $M$ with values in $E$ and with Dirichlet boundary conditions at the boundary.

Assume that $g$ is product near $N_1$ and $N_2$ and define a smooth Riemannian manifold $\widetilde{M}$ by gluing a second copy of $M$ along $N_1 \sqcup N_2$. Note that $\partial \widetilde{M}$ consists of two copies of $N_3$.

Let $N_1$ and $N_2$
be identified by an isometry $f$, and denote by $M_f$ the Riemannian manifold obtained
from $M$ by gluing $N_1$ onto $N_2$. The flat Hermitian vector bundle $(E,h, \A)$
induces smooth flat vector bundles over $M_f$ and $\widetilde{M}$, since $h$ and $\A$ 
are product near $\partial M$. We write $\Delta_{\widetilde{M}}$ and 
$\Delta_{M_f}$ for the twisted Laplacians on 
$\widetilde{M}$ and $M_f$, respectively, with Dirichlet boundary 
conditions at the respective boundaries.

Consider, as in \S \ref{int-2}, a simplicial complex $K$ which triangulates $M$ with
subcomplexes $L_1, L_2$ and $L_3$ triangulating $N_1, N_2$ and $N_3$, respectively.
Its double $\widetilde{K}$ along $L_1 \sqcup L_2$, with the boundary subcomplex $L_3 \sqcup L_3$ is a simplicial decomposition of $\widetilde{M}$.
The simplicial complex $K_f$, obtained by gluing $K$ along the two identified
boundary components $L_1$ and $L_2$, decomposes $M_f$.

The pullbacks of the combinatorial analog of $E$ over $K$ define combinatorial vector bundles
over $K_f$ and $\widetilde{K}$. The combinatorial Riemannian structure on $K$ defines natural
combinatorial Riemannian structures on $K_f$ and $\widetilde{K}$. The metric structure on $M$
and the Whitney map an define a combinatorial Riemannian structure on $K$ together with
the inner product on corresponding cochains with
values in combinatorial vector bundle.
Denote by $\Delta_{K_f}$ and $\Delta_{\widetilde{K}}$ the combinatorial
Laplace operators on cochains of degree zero, defined with respect to
this inner product, with Dirichlet boundary conditions at the
boundary. In \S \ref{approximation} we prove the following

\begin{thm}\label{main2}
As the mesh $\delta>0$ of the triangulation $K$ goes to zero
under standard subdivisions\footnote{Standard subdivisions were introduced by
Whitney in \cite{Whi:GIT}}
\begin{equation}
 \begin{split}
\lim_{\delta \to 0}  \frac{(\det' \Delta_{K_f})^2}
{\det' \Delta_{\widetilde{K}}}
= \frac{(\det\nolimits_\zeta \Delta_{M_f})^2}
{\det\nolimits_\zeta \Delta_{\widetilde{M}}}
 \end{split}.
\end{equation}
\end{thm}

By an application of Theorem \ref{main1}, \ref{main2}
and the gluing formula of Burghelea, Friedlander and Kappeler,
which we recap in \S \ref{BFK}, we arrive at the following

\begin{cor}\label{main3}
Let $\mathscr{R}_a(M_f,N_2)$ be the analytic Dirichlet-to-Neumann map and $\mathscr{C}_{M_f,N_2}$ the
corresponding constants\footnote{This constant has been explicitly identified in \cite{Lee:BFK}.} in the analytic Mayer-Vietoris formula.
Then, as the mesh $\delta>0$ of the triangulation $K$ goes to zero
under standard subdivisions
\begin{equation}
 \begin{split}
\lim_{\delta \to 0} \, &\frac{(\det'\mathscr{R}_c(K_f,L_2))^2}
{\det'\mathscr{R}_c(\widetilde{K},L_1 \sqcup L_2)}
= \frac{\mathscr{C}^2_{M_f,N_2}}{\mathscr{C}_{\widetilde{M},N_1\sqcup N_2}}
\cdot \frac{ (\det\nolimits_\zeta \mathscr{R}_a(M_f,N_2))^2}
{\det\nolimits_\zeta \mathscr{R}_a(\widetilde{M},N_1\sqcup N_2)}
   \end{split}.
\end{equation}
\end{cor}

The presented statements hold without the assumption of orientability for $K$ and $M$.
Still, orientability is necessary for the discussion of combinatorial covariant derivatives in \S
\ref{comb} as well as in the general setup of the quantum field theoretic framework in \S \ref{framework}.

The paper is structured as follows. We first provide in \S \ref{BFK} an overview over the fundamental
elements of spectral geometry and the Mayer-Vietoris formula for zeta-determinants by
Burghelea, Friedlander and Kappeler in the special case of scalar Laplace-Beltrami operators.
With continue in \S \ref{comb} with a detailed construction of the combinatorial vector bundles, connections
and discrete Laplacians, for the moment independent from the possibly underlying Riemannian
geometry. In \S \ref{Dodziuk} we discuss the approximation theory of Dodziuk.
\S \ref{DN} introduces the classical scalar free theory.
Its critical point defines a Poisson and subsequently the Dirichlet-to-Neumann operators.
We define the combinatorial Gaussian quantum field theory and prove the combinatorial
gluing formula in \S \ref{QFT}. The last \S \ref{approximation} establishes a link between the
determinant gluing identities for discrete Laplacians and the Mayer-Vietoris formula by Burghelea,
Friedlander and Kappeler. We conclude with an outlook of various research directions which are interesting
in relation to the present analysis.

{\bf Acknowledgements.} Both authors would like to
thank Alberto Cattaneo, Theo Johnson-Freyd, John Lott, Rafe Mazzeo, Pavel Mnev, Peter Teichner and Ananth Sridhar for important discussions at various stages of this work, with special thanks
to Werner M\"uller and Dennis Sullivan for valuable input. 
B.V. also thanks University of California at Berkeley for hospitality and
gratefully acknowledges financial support of the Hausdorff
Research Institute of Mathematics in Bonn. N.R. acknowledges support from the Chern-Simons chair, from
the NSF grant DMS-1201391 and the University of Amsterdam where important part of the work
was done. The work was completed when both authors attended the thematic program "Modern Trends in TQFT"
at the Erwin Schr\"odinger Institute in Vienna.

\tableofcontents

\section{Analytic Laplacians and Mayer-Vietoris formula \\ 
for their $\zeta$-regularized determinants}
\label{BFK}
Here we will recall basic facts about differential forms on Riemannian manifolds,
Laplace operators, and Mayer-Vietoris type formulae for $\zeta$-regularized determinants
of Laplace operators. We assume orientability in some subsections for simplicity. For non-orientable
spaces the arguments carry over after twisting with the orientation density bundle.

\subsection{Differential forms on Riemannian manifolds}

\subsubsection{The twisted de Rham complex}
Let $M$ be a compact smooth manifold with boundary $\partial M$.
Let $E$ be a Hermitian vector bundle over $M$ with finite dimensional fibers, with Hermitian metric $h$ and a flat connection $A$ on it. Define $\Omega^\bu(M,E)$ as the space of $E$-valued
differential forms and  $\Omega^\bu_c(M,E)$ as the subspace of smooth differential forms with
compact support away from $\partial M$. Let $(\Omega^\bu_c(M,E),d)$ be the twisted
de Rham complex, where $d$ is the twisted (by the flat connection $A$) de Rham differential
$$
d_q:\Omega^{q}_c(M,E)\to \Omega^{q+1}_c(M,E).
$$
Locally, over a neighborhood $U$, the connection $A$ is a 1-form with coefficients in $End(V)$
and the twisted de Rham differential acts on forms as
$$
d\omega=d_{dR}\omega+A\wedge \omega,
$$
where $d_{dR}$ is the de Rham differential acting of $\Omega^\bu(U, V)$. 
The flatness of $A$ means $d^2=0$.
Locally, this is equivalent to $d_{dR}A+A\wedge A=0$. Each summand
in the right hand side of this formula, in general, is defined only locally. The twisted differential
however, is defined globally and this is why we need a flat connection on $E$ if we want to have globally defined differential on $\Omega^\bullet(M,E)$.

\subsubsection{The scalar product}
Now assume $(M,g)$ is an oriented Riemanian manifold of dimension $n$.
Let $*: \Omega^q(M,E)\to \Omega^{n-q}(M,E)$ be the Hodge star operator,
which yields a natural scalar product on the space $\Omega^\bu(M,E)$
\begin{equation}\label{sc-prod}
\langle\omega, \omega'\rangle_M=\int_M (\omega,\wedge *\omega'),
\end{equation}
where $(\cdot , \cdot)$ is the fiberwise Hermitian product.

With respect to the Riemannian metric on $M$, each $q-$form $\w \in \Omega^\bu(M,E)$
has a natural decomposition into its normal and tangential components
\footnote{ If $\iota: \pa M\to M$ is the natural inclusion of the boundary mapping,
we have $\omega_{\textup{tan}}=\iota^*(\omega)$ and $\w_{\textup{norm}}=*_\pa\iota^*(*\w)$
where $*$ is the Hodge star operation on $M$ and $*_\pa$ is the Hodge star operation on $\pa M$
corresponding to the metric induced from $M$. Let $\nu$ be a unit normal vector field on $\pa M$
which is positive with respect to the orientation on $M$. Then $\omega_{norm}=\iota^*(\iota_\nu\omega)$
where $\iota_\nu\omega$ is the contraction of $\nu$ and $\omega$.}
$
\w = \nu^* \wedge \w_{\textup{norm}} + \w_{\textup{tan}}
$
at any point of the boundary $\partial M$.

The Stokes formula applied to $d(\omega\wedge *\omega')$ gives the integration by parts formula
\begin{align}\label{stokes1}
\langle d \omega,  \omega' \rangle_M = \langle \omega, d^\vee \omega' \rangle_M + \langle \omega_{\tan}, \omega'_{\textup{norm}} \rangle_{\pa M}.
\end{align}
Here $d^\vee$ is the formal adjoint of $d$ with respect to the scalar product (\ref{sc-prod}).
In terms of the Hodge $*$ operation
$d^\vee\omega=(-1)^{nq+n+1}*d*\omega$ when $\omega\in \Omega^q(M,E)$.

\subsubsection{Laplace operators}
Here we recall the basic notions related to Laplace operators. Laplacians are defined by
\[
\Delta^q=d_{q}d^\vee_{q+1}+d^\vee_{q}d_{q-1}:
\Omega^q_c(M,E) \to \Omega^q_c(M,E).
\]

We regard $d_q$ and $\Delta^q$ as unbounded operators in
$L^2\Omega^\bu(M,E)$ with domain $\Omega^\bu_c(M,E)$.
Recall that the maximal extension $d_{q,\max}$ of $d_q$ is a linear operator on
$L^2\Omega^q(M,E)$ with domain
$$
\dom (d_{q,\max})= \{\w \in L^2\Omega^q(M,E) \mid d_q \w \in L^2\Omega^{q+1}(M,E)\},
$$
This is the space of forms $\omega$ from $L^2\Omega^q(M,E)$  such that the differential $d_q \w$
is not just a distribution but actually a form in
$L^2\Omega^{q+1}(M,E)$.

The minimal extension $d_{q,\min}$ of $d_q$ with domain
$\dom (d_{q,\min}) \subset \dom (d_{q,\max})$ is the
graph closure\footnote{Recall that the graph closure in our case is the
closure with respect to the metric $||\omega||^2=\langle\omega,\omega\rangle+
\langle d\omega, d\omega\rangle$.}
of $d_q$ on $\Omega^q_c(M,E)$. Ideal boundary conditions for the
de Rham complex $(\Omega^\bu(M,E),d)$ is a choice
of closed extensions $D_q$ of $d_q$ for each $q=0,...,\dim M$ with
\[
\dom (d_{q,\min}) \subseteq \dom (D_q) \subseteq \dom (d_{q,\max}),
\]
such that $D_q:\dom (D_q) \to \dom (D_{q+1})$
and $D_{q+1} \circ D_q =0$. Such boundary conditions combine into a Hilbert complex in the sense of
\cite{BruLes:HK}. Ideal boundary conditions for the de Rham complex
induce a self-adjoint extension $(D^*_qD_q + D_{q-1}D_{q-1}^*)$ for each $\Delta^q$.

Two special cases of relative and absolute boundary conditions correspond
to minimal and maximal extensions of $d$
\begin{equation}
\begin{split}
\Delta^q_{\textup{rel}} &=
d^*_{q,\min}d_{q,\min} + d_{q-1,\min}d^*_{q-1,\min}, \\
\Delta^q_{\textup{abs}} &=
d^*_{q,\max}d_{q,\max} + d_{q-1,\max}d^*_{q-1,\max}. \\
\end{split}
\end{equation}
Explicitly, these boundary conditions are given as follows.
The \emph{relative self-adjoint extension} $\Delta^q_{\textup{rel}}$
can be defined as the closure in $L^2\Omega^q(M,E)$ of the action
of $\Delta^q$ on forms satisfying \emph{relative} (or Dirichlet) boundary conditions
\begin{align*}
\w_{\textup{tan}}=0 \ \textup{and}
 \ (d^\vee_{q-1} \w)_{\textup{tan}}=0 \ \textup{at}  \ \partial M.
\end{align*}

The \emph{absolute self-adjojnt extension} $\Delta^q_{\textup{abs}}$
is given by the closure in $L^2\Omega^q(M,E)$ of the action
of $\Delta^q$ on forms satisfying \emph{absolute} (or Neumann) boundary conditions
$\w_{\textup{norm}}=0$ and $(d_{q} \w)_{\textup{norm}}=0$ at $\partial M$.

In this paper we will focus on Laplace Beltrami operators
with Dirichlet (relative) boundary conditions in degree $q=0$.

\subsection{Dirichlet-to-Neumann operator}
Let $M$ be a Riemannian, smooth, oriented manifold with boundary $\pa M$.
Denote by $\Delta_M$ the Hodge Laplace operator on $\Omega^0(M,E)=C^\infty(M,E)$.
It is well known that for each $\eta\in C^\infty(\partial M,E)$ the Dirichlet boundary problem
\begin{align}
\Delta_M\phi = 0, \quad \phi|_N=\eta.
\end{align}
has unique solution.  Denote by $P_M: C^\infty(\pa M; E)\to C^\infty(M; E)$
the corresponding Poisson operator. In terms of $P_M$ the unique solution
to the Dirichlet problem is $\phi=P_M\eta$.

The Dirichlet-to-Neumann operator $R^M$ may be defined implicitly as a linear map
such that for any $\eta, \eta' \in C^\infty(\partial M,E)$
\begin{align}\label{RM}
\langle d P_M \eta , d P_M \eta' \rangle = \langle R^M \eta, \eta' \rangle_{L^2(\partial M, E)}.
\end{align}
Explicitly, this leads by \eqref{stokes1} to the following expression
\[
R^M=\pa_\nu P_M: C^\infty(\pa M,E)\to C^\infty(\pa M,E),
\]
where for any $\psi \in C^\infty(M,E)$ we write $\pa_\nu \psi := \iota^* (\iota_\nu d \psi)\in C^\infty(\pa M,E)$,
with $\iota: \pa M \hookrightarrow M$ denoting the natural inclusion, and $\iota_\nu$ the contraction with the unit
normal vector field on $\pa M$.

Recall that it is called the Dirichlet-to-Neumann map because the solution to the
Neumann problem $\Delta_M\psi=0, \ \ \pa_\nu \psi=\xi$ can be written
as $\psi=P_M\eta$ where $\xi=R^M\eta$. In other words, $R^M$ maps Dirichlet boundary data to
the Neumann one.

\subsection{$\zeta$-regularized determinants}
Fix Dirichlet boundary conditions for $\Delta_M$, and denote the resulting
self-adjoint extension again by the same letter for the moment.
It is known that $\Delta_M$ has non-negative pure point spectrum and  that for every $t>0$,
$\exp(-t\Delta_M)$ is a trace class operator with an asymptotic expansion
\begin{align}\label{heatasymp}
\textup{Tr}\, \left(e^{-t\Delta_M} \right)\sim
\sum_{j=0}^{\infty}A_k (\sqrt{t})^{-\dim M+j}, \quad t\to 0+.
\end{align}

Let $\{\lambda_j\}_{j\in \N}$ denote the eigenvalues of $\Delta_M$.
It follows from \eqref{heatasymp} that Weyl's law holds for the counting
function of the eigenvalues. This implies that the zeta function
\[
\zeta(s;\Delta_M):=\sum_{\lambda_j \neq 0}\lambda_j^{-s}
\]
converges in the half-plane $\Re(s)>\dim M/2$ and it can be expressed in terms of
the trace of the heat operator by
\begin{align}
\zeta(s, \Delta_M)=\frac{1}{\Gamma(s)}\int_0^{\infty}t^{s-1}
\textup{Tr}(e^{-t\Delta_q}-P_q) \;dt, \quad \textup{Re}(s)>\dim M/2,
\end{align}
where $P_q$ denotes the orthogonal projection
of $L^2(M,E)$ onto $\ker \Delta_M$.
Then the asymptotic expansion \eqref{heatasymp} yields the meromorphic
extension of the right hand side and hence, of the zeta function, to the
whole complex plane $\C$. Furthermore it also follows from \eqref{heatasymp} that
$\zeta (s, \Delta_M)$ is regular at $s=0$. Hence we can define the
$\zeta$-regularized determinant of $\Delta_M$ by the following expression
\begin{align*}
 \det\nolimits_\zeta \Delta_M := \exp \left( - \left. \frac{d}{ds}\right|_{s=0}
  \zeta(s,\Delta_M)\right).
\end{align*}
This discussion extends to other discrete operators with spectrum bounded
from below, trace class heat operator and an asymptotic expansion
of the form \eqref{heatasymp}. In the finite dimensional case this defines exactly
the determinant of a given self-adjoint linear operator.

\subsection{Mayer-Vietoris type formula for $\zeta$-regularized determinants}

Let $(M,g)$ be a compact Riemannian manifold
with boundary $\partial M$ that consists of three disjoint boundary components
$N_1,N_2$ and $N_3$. Consider a flat Hermitian vector bundle $(E,h, \A)$, i.e. 
Hermitian bundle with a Hermitian flat connection $\A$. Flatness implies product 
structure over a collar neighborhood of $\partial M$. We denote by $\Delta_{M}$ 
the Laplace Beltrami operator acting of functions on $M$ with values in $E$ and 
with Dirichlet boundary conditions at the boundary.

Assume that $g$ is product near $N_1$ and $N_2$. Let $N_1$ and $N_2$
be identified by an isometry $f$, and denote by $M_f$ the Riemannian manifold obtained
from $M$ by gluing $N_1$ onto $N_2$. The flat Hermitian vector bundle $(E,h, \A)$
induces a smooth flat vector bundle over $M_f$.
We write $\Delta_{M_f}$ for the twisted Laplacian on $M_f$, with Dirichlet boundary conditions at $N_3$\footnote{Dirichlet
boundary conditions at $N_3$ may be replaced be any elliptic boundary conditions for $\Delta_{M}$ and $\Delta_{M_f}$.}.

Let $P_M$ be the Poisson operator solving the Dirichlet problem on $M$.
The Dirichlet-to-Neumann operator is this setting differs from \eqref{RM} and is defined as
\[
\mathscr{R}_a(M_f,N_2)=(\pa_{\nu_1} + \pa_{\nu_2}) P_M.
\]
Here $\pa_{\nu_i},i=1,2,$ denotes the normal derivative at $N_i$ in the inward direction.
The operators $\mathscr{R}_a(M_f,N_2)$ and $R^M$ are related by the following
identity
$$\langle \mathscr{R}_a(M_f,N_2) \eta, \eta \rangle_{N_2} =
\langle R^M (f^*\eta, \eta), (f^*\eta, \eta) \rangle_{N_1 \sqcup N_2}.$$
\begin{thm}\label{BFK-theorem}\cite{BFK:MVT}
There exists an explicitly determined constant $\mathscr{C}_{M_f,N_2}\in \R$ such that
$$\frac{\det\nolimits_\zeta \Delta_{M_f}}
{\det\nolimits_\zeta \Delta_{M}}
= \mathscr{C}_{M_f,N_2} \det\nolimits_\zeta \mathscr{R}_a(M_f,N_2).$$
\end{thm}

This is an example of Mayer-Vietoris type formulae obtained by Burghelea, Friedlander and Kappeler
for elliptic differential operators of possibly higher order.
An explicit formula for $\mathscr{C}_{M_f,N_2}$ has been obtained
by Yoonweon Lee \cite{Lee:BFK}.

\section{Combinatorial vector bundles and Laplacians} \label{comb}

In the present section we provide a purely combinatorial definition 
of vector bundles over simplicial complexes and introduce a notion of a
(combinatorial) connection, its curvature and the corresponding 
covariant derivative. These constructions lead to the definition of 
a combinatorial Hodge Laplacian, where the classical discretization of the 
analytic Hodge Laplacian, see \cite{Dod:FDA} and \cite{Mue:ATA}, is included as a 
special case. 

The main novelty of the presented discussion is the introduction
of the necessary combinatorial concepts without reference to the Riemannian geometry,
in contrast to the classical references \cite{DodPat}, \cite{Whi:GIT}

\subsection{Vector bundles and connections on a simplicial complex}

\subsubsection{Vector bundles} Throughout this section we assume that
$K$ an oriented $n$-dimensional simplicial complex.
We denote the set of $q$-simplices by $K_q$ and for each $\sigma \in K_q$ we will write $\partial \sigma$
for $(q-1)$-simplices forming its boundary.
\begin{definition} A vector bundle $E$ over a simplicial complex $K$ with the
fibre $V$ is a triple $(E, \pi, K)$ where $E$ is the total space of the bundle, $\pi: E\to K$
is a projection such that for each simplex $\sigma\in K$
\[
\quad E_{\sigma} = \pi^{-1}(\sigma) \cong V.
\]
as a vector space. The space $E_\sigma$ is called the fiber of $\pi$ over $\sigma$.
\end{definition}

Any vector bundle over a simplicial complex $K$ is trivializable: a choice of a linear isomorphism $\xi_\sigma: E_\sigma\simeq V$
for each $\sigma\in K$ induces an isomorphism $E\simeq V\times K$ and brings $\pi$ to the
natural projection $V\times K\to K$.

As usual, a section of $E$ is a map $s: K\to E$, such that $s\cdot\pi=\textup{id}_K$.
Denote the space of sections by $C^\bullet(K,E)$.
A trivialization of $E$ identifies the  space of sections $C^\bullet(K,E)$ with the space of maps $K\to V$.
The space $C^\bullet(K,E)$  is a vector space with $(\phi_1+\phi_2)(\sigma)=
\phi_1(\sigma)+\phi_2(\sigma) \in E_\sigma$. It has a natural $\Z$-grading with $C^q(K,E)$
being the space of sections $\phi: K_q\to E$ and
co-chain complex structure with the differential
\begin{align}
\left( d \phi \right) (\tau) = \sum_{\sigma \in \partial \tau}
(-1)^{(\tau, \sigma)}  \phi(\sigma),
\end{align}
Here $\tau\in K_{q+1}$ and $(-1)^{(\tau, \sigma)}$ is plus when orientations of $\sigma$ and
$\tau$ agree and it is minus when they are opposite. The identity $d^2=0$ is proven below.

This cochain complex should be regarded as a discrete version of the
de Rham complex on a trivial vector bundle with the trivial flat connection.

\subsubsection{The double of a simplicial complex}
Let $K^\vee$ be the complex dual to $K$. By definition, there is
a bijective correspondence $j: K_q \to K^\vee_{n-q}$. Moreover,
the boundary operator $\partial$ on $K$ defines a coboundary mapping
$\pa^\vee$ on $K^\vee$, such that for each $\sigma \in K_q$ and the cochain $(j\sigma)^*$
dual to $j\sigma \in K^\vee_{n-q}$, we find  $$\pa^\vee (j\sigma)^* = (j\partial \sigma)^*.$$

Geometrically, $0$-simplices of $K^\vee$ can be identified
with a choice of a point in the interior of each $n$-simplex of $K$ (we will call such points
centers of $n$-simplices), $1$-simplices of $K^\vee$
can be identified with a choice of a point on each $(n-1)$-simplex of $K$ (centers of $(n-1)$-simplices) and connecting
it with the center of each adjacent $n$-face by a segment, etc. We will denote a center of $\sigma$ by $p_\sigma$.

Clearly, a particular choice of the centers $p_\sigma$ for each $\sigma$
is irrelevant in the combinatorial picture and becomes of interest only
when the simplicial complex is used to approximate smooth structures.

Define the simplicial complex $D(K)$ which we will call the \emph{double of} $K$ as follows.
For each $\sigma\in K$ choose a center $p_\sigma\in \textup{int}(\sigma)$ of $\sigma$.
Add new edges to $K$ which connect the center of $\sigma$
with centers of $\tau\in \pa \sigma$, i.e. $p_\sigma \in K^\vee$. Subdivide each $2$-simplex by new edges connecting its
center to centers of boundary edges. For each 3-simplex add 2-simplices whose boundary
consists of new edges connecting the center of the 3-simplex to centers of its boundary 2-simplices and
new edges connecting centers of boundary 2-simplices to centers of their boundary etc. Proceeding
iteratively in higher degrees defines the double simplicial complex $D(K)$.

Note for example that by construction, each such 2-simplex is a quadrilateral with one vertex
$p_\eta$ being the center of some $\eta\in K_3$, two others $p_{\tau^\pm}$ being
centers of $\tau^\pm\in \pa \eta$, and the other $p_\sigma$ where $\sigma\in \pa\tau^\pm$.
Moreover all simplices of $D(K)$ are cubic and thus, $D(K)$ is an example of a cubic simplicial complex.

Note that vertices ($0$-simplices) of $D(K)$ are in natural bijection with simplices of $K$.
An example of the double of a 2-dimensional simplicial complex is shown on Fig. 1,
where the simplices of the double complex are bold.

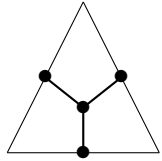
\begin{figure}[h]
\begin{center}
\begin{tikzpicture}
\draw (-1,0) -- (1,0);
\draw(-1,0) -- (0,2);
\draw (1,0) -- (0,2);

\node at (0,0) {$\bullet$};
\node at (-0.5,1) {$\bullet$};
\node at (0.5,1) {$\bullet$};
\node at (0,0.6) {$\bullet$};

\draw[thick] (0,0.6) -- (0,0);
\draw[thick] (0,0.6) -- (0.5,1);
\draw[thick] (0,0.6) -- (-0.5,1);

\end{tikzpicture}
\end{center}
\label{double}
\caption{The double simplicial complex $D(K)$ of a triangle.}
\end{figure}

\subsubsection{Connections on vector bundles} Fix a bijection between simplices of $K$ and vertices of
$D(K)$. Then a vector bundle over $K$ becomes a vector bundle over the set of vertices of $D(K)$.
Naturally any such vector bundle is trivial and isomorphic to $V(D(K))\times V$ where $V(D(K))$
is the set of vertices of $D(K)$ and $V$ is the vector space isomorphic to a fiber.

\begin{definition} A connection $\A$ on a vector bundle $E$ over $K$ is the collection of \emph{parallel transports},
i.e. it assigns a linear isomorphism  of fibers $\A (\gamma)\equiv \A(\tau, \sigma): 
E_\sigma \to E_\tau $ to each edge path $\gamma\subset D(K)$ connecting 
simplices $\tau, \sigma \in K$.
\end{definition}

\begin{definition} The vector bundle $E$ is Hermitian if each fiber is equipped with
a Hermitian scalar product.
A Hermitian connection is the collection of parallel transports which are
unitary isomorphisms of fibers.
\end{definition}

\begin{definition} Given a connection $\alpha$, its combinatorial \emph{covariant derivative}
$d_\A: C^q(K,E) \to C^{q+1}(K,E)$
is:
\begin{align}
\left( d_\A \phi \right) (\tau) = \sum_{\sigma \in \partial \tau}
(-1)^{(\tau, \sigma)}\A(\tau, \sigma) \phi(\sigma).
\end{align}
Here $\phi \in C^q(K,E)$, $\tau\in K^{q+1}$, $(\tau, \sigma)=0$, if the orientation of
$\sigma$ coincides with the orientation induced by $\tau$ on its boundary
and $(\tau, \sigma)=1$ if these orientations are opposite.
\end{definition}

Consider four simplices $\eta\in K_{q+2}$, $\tau_\pm\in K_{q+1}$ and $\sigma \in K_q$
such that $\tau_\pm\in \pa \eta$ and $\sigma\in \pa \tau_\pm$. For given $\eta$ we will call such
such quadruples $(\eta, \tau_\pm, \sigma)$ \emph{leaves} of $\eta$ and will denote the set of \emph{leaves of} $\eta$
by $L(\eta)$. We identify each leaf with the corresponding $\sigma \in \partial \partial \eta$
and write $\sigma \in L(\eta)$. A choice of such $(\sigma, \eta)$ determines the associated
simplices $\tau^\pm$. Let us compute $d_\A^2$. For  $\phi \in C^q(K,E)$ and $\eta \in K_{q+2}$ we have
\begin{align*}
(d_\A d_\A \phi) (\eta) &= \sum_{\tau \in \partial \eta}
\sum_{\sigma \in \partial \tau} (-1)^{(\tau, \eta) + (\sigma, \tau)}
\A (\eta, \tau)  \A (\tau, \sigma) \phi(\sigma)
\\ &= \sum_{\sigma\in L(\eta)} \left(
(-1)^{(\tau^+, \eta) + (\sigma, \tau^+)}
\A  (\eta, \tau^+) \A ( \tau^+, \sigma)\right. \\
&+ \left. (-1)^{(\tau^-, \eta) + (\sigma, \tau^-)}
\A  (\eta, \tau^-)  \A ( \tau^-, \sigma)\right)
\phi(\sigma)\\ &=\sum_{\sigma\in L(\eta)} (-1)^{(\sigma,\eta)}F(\sigma,\eta)\phi(\sigma),
\end{align*}
where we introduce the \emph{curvature} of a connection as the collection of elements
\begin{equation}\label{curva}
F_\A(\sigma, \eta)=\A  (\eta, \tau^+)  \A ( \tau^+,\sigma)
- \A  (\eta,\tau^-)  \A ( \tau^-,\sigma)
\end{equation}
assigned to each quadrilateral of $D(k)$.
In the last step of the computation we used the identity
\begin{align*}
(-1)^{(\tau^+, \eta) + (\sigma, \tau^+)} +
 (-1)^{(\tau^-, \eta) + (\sigma, \tau^-)} = 0.
\end{align*}

\begin{definition} The  \emph{curvature of} connection $\A$ on the leaf of $\eta$ containing $\sigma$
is defined by $F_\A(\sigma, \eta) \in \textup{Hom}(E_\sigma, E_\eta)$.
\end{definition}

\begin{defn}\label{flatness}
The combinatorial connection $\A$ is \emph{flat} if its curvature is zero for
each pair $(\sigma\in L(\eta), \eta)$, i.e. if
\[
\A  (\eta, \tau^+)  \A ( \tau^+,\sigma)
= \A  (\eta,\tau^-)  \A ( \tau^-,\sigma)
\]
for each such pair $(\sigma, \eta)$.
\end{defn}

The following is now clear:
\begin{prop}
The combinatorial connection is flat, if and only if $$d_\A^2= 0.$$
\end{prop}

One of the corollaries  is the identity $d^2=0$ for the non-twisted differential.
A flat combinatorial connection $\A$ defines
a cochain complex $(C^*(K,E), d_\A)$. This complex should be regarded as a
discrete analog of the de Rham complex with coefficients in a local system on $E$.

\subsection{Combinatorial Laplacians.}

\subsubsection{Riemannian structure on simplicial complexes}

A choice of metric on an oriented Riemannian manifold $M$ defines a scalar product on forms, i.e.
an $L_2$-norm on the space of forms. It also defines a scalar product and a norm on forms
on each submanifold of $M$. We will call such norm Riemannian norms.

\begin{definition} A Riemannian norm on an $n$-dimensional simplicial complex $K$ is a mapping
$|\cdot |: K_q\times K_q \to \RR$ for each $q=0,1,\dots, n$, such that the corresponding
bilinear form on $C_*(K, \R)$ is positive definite. The number $|\sigma, \tau|$ will
is called the weight of the pair. A simplicial complex with a Riemannian norm
is called a \emph{metrized} simplicial complex.
\end{definition}

\begin{remark} An embedding of the simplicial complex $K$ to a Riemannian manifold $M$, gives a natural choice
of a Riemannian norm on $K$. It is induced by the metric on $M$ and
by the Whitney forms (see section \ref{WM}). It assigns to each pair $(\sigma, \tau)\in K_q\times K_q$
the weight
\[
|\sigma, \tau|=\int_{M} W(\sigma)\wedge \ast W(\tau).
\]
\end{remark}

\begin{remark} Note that if we want to define in a similar way a discrete analog of metric
on $M$, we should define a similar systems of weights $|\sigma, \tau|_L$ for every subcomplex
$L\subset K$. It is natural to require that the scalar products of
cochains on sub-complexes should
satisfy the natural compatibility condition. If two subcomplexes
$L_1$ and $L_2$ of the same dimension contain a subcomplex $L_3$
of the same or lower dimension, then the systems of weights for $L_1$ and $L_2$
agree on $L_3$ and coincide with the system of weights for $L_3$ induced from $K$. 

\end{remark}

Define a Hermitian metric on a vector bundle $E$
over $K$ as a scalar product $h_\sigma:E_\sigma \times E_\sigma \to \R^+$ on fibres $E_\sigma$.
Such a combinatorial vector bundle, equipped with a Hermitian metric is called the
\emph{Hermitian} vector bundle.

For  a Hermitian vector bundle  $E$ over a metrized simplicial complex $K$ define
the bilinear form $\langle .,.\rangle_K $ on $C^q(K,E)$ as
\begin{align}\label{inner}
\langle \phi, \psi \rangle_K =
\sum_{\sigma \in K_q} \sum_{\tau \in K_q}
h(\phi(\sigma), \A(\sigma, \tau)\psi(\tau)) |\sigma, \tau|.
\end{align}
We also require that this form is positive definite, i.e. that it is a scalar
product on $C^q(K,E)$. The scalar product with $|\sigma, \tau| = |\sigma, \sigma| \delta_{\sigma, \tau}$
is called \emph{diagonal}. In case of a diagonal inner product,
$d_\A^\vee$ can be computed explicitly
\[
(d_\A^\vee \phi)(\sigma) = \sum_{\{\tau \mid \sigma \in \partial \tau\}}
(-1)^{(\tau, \sigma)} \A(\sigma, \tau) \phi(\tau)
\frac{|(\tau, \tau)|}{|\sigma, \sigma|}.
\]
However, scalar products need not be diagonal in general, 
which is highlighted by the presence of holonomy $\A$ in \eqref{inner}. 
For example, one of the geometrically most natural scalar products, the Whitney product is
not diagonal, in which case we recover the general formula \eqref{inner}, 
see \S \ref{whitney-subsubsection} below. 

Recall that a connection $\alpha$ is {\it Hermitian} if for any $v\in E_\sigma$
and $w\in E_\tau$ the parallel transport is given by unitary transformations\footnote{Here we will focus on
real scalar product, i.e. $h(\phi(\sigma), \psi(\sigma))=h(\psi(\sigma), \phi(\sigma))$. In this case
a connection is given by a collection of $h$-orthogonal parallel transports. In the complex Hermitian case
when parallel transport is given by $h$-unitary matrices $h(\phi(\sigma), \psi(\sigma))=\overline{h(\psi(\sigma), \phi(\sigma))}$. 
We will use the term Hermitian for both complex Hermitian and real orthogonal cases.}:
\[
h(\A(\tau, \sigma)v, w) = h(v, \A(\tau, \sigma)^{-1}w).
\]
By definition, a parallel transport along the edge oriented from $\sigma$ to $\tau$ is the inverse
to the one oriented from $\tau$ and $\sigma$: $\A(\sigma, \tau)=\A(\tau, \sigma)^{-1}$. For a Hermitian connection 
$\alpha$ the inner product on $C^\bullet(K,E)$ defines the adjoint covariant derivative
$d^\vee_\alpha: C^q(K,E)\to C^{q-1}(K,E)$ for any $\phi\in C^q(K,E)$ and $\psi \in C^{q-1}(K,E)$, as
\[
\langle d_\A^\vee \phi, \psi \rangle_K := \langle \phi, d_\A \psi \rangle_K.
\]

The associated combinatorial Laplacian is defined as
\begin{align}\label{lapl}
\Delta^K_\A:= d_\A  d_\A^\vee + d_\A^\vee  d_\A: C^q(K,E) \to C^q(K,E).
\end{align}
Note that twisted combinatorial Laplacian are defined for all connections $\A$.

If $\alpha$ is flat, i.e. $d_\alpha^2=0$, then $(d_\alpha^\vee)^2=0$. Such two operators together
with the scalar product give the discrete version of the Hodge structure.  We will
also use the notation $\Delta_K$, whenever the connection $\A$ is fixed.

\subsubsection{Locality}\label{locality} We will say that the Riemannian norm on an $n$-dimensional simplicial complex $K$
is {\it local} if the weight $|\sigma, \tau|$ in the scalar product of $q$-cochains, $q=0,1, \dots, n$ can be non-zero only if
\begin{itemize}
\item there exists $\eta\in K_{q+1}$ such that $\sigma,\tau\subset \pa \eta$,
\item the subcomplex $\overline{\sigma\cup \tau}$ is connected.
\end{itemize}
The locality of the scalar product is equivalent to a Mayer-Vietoris property
for any subcomplex $L$ that consists of three connected components
$L= L_1 \sqcup L_2 \sqcup L_3$, where $L_1$ and $L_2$ are isometrically identified via $f$.
This defines a new chain complex $K_f$ with a Riemannian norm induced from $K$.
Let $\phi, \psi$ be cochains on $K$, such that $\phi|_{L_1}=f^*\phi|_{L_2}$ and $\psi|_{L_1}=f^*\psi|_{L_2}$.
Then the cochains lift to well-defined cochains $\phi_f, \psi_f$ on $K_f$ and the
Mayer-Vietoris property is given by\footnote{The pairing on $L_1$ is defined by restriction from $K$.}
\[
\langle \phi , \psi \rangle_K = \langle \phi_f , \psi_f \rangle_{K_f} + \langle \phi|_{L_1} , \psi|_{L_1} \rangle_{L_1}.
\]

For vertices this means $|\sigma,\tau|\neq 0$ only when $\tau$ is connected to $\sigma$ by an edge.
For edges this means that $\sigma$ and $\tau $ belong to the boundary of some $2$-simplex and
that they share a vertex.

The Whitney scalar product, which we introduce below in 
\S \ref{whitney-subsubsection}, has minimal possible locality in the sense that 
$|\sigma,\tau|$ is non-trivial on all vertices $\sigma$ and $\tau$ connected by an edge. The diagonal scalar product
has maximal locality in the sense that $|\sigma,\tau|$ is trivial whenever $\sigma \neq \tau$.

One of the reasons we define this notion of locality is that it is consistent
with the locality of the classical action for scalar field theory.

\subsubsection{Green's formula.} Assume that an $n$-dimensional simplicial complex 
$K$ contains a subcomplex $L$ with $\dim(L)=n-1$ such that to each $n-1$ simplex
in $L$ belongs to the boundary of only one $n$-dimensional
simplex in $K$. We will call $L$ the \emph{boundary} of $K$. We have a natural 
projection $p: C^\bullet(K,E)\to C^\bullet(L,E)$
which is the restriction of a cochain on $K$ to $L$.

Define co-chains $C^\bullet(K,L,E)$ as
the $\ker(p)\subset C^\bullet(K,E)$. These cochains are discrete analog of differential forms
on a manifold with vanishing pullback to the boundary, i.e. of differential forms producing
relative cohomologies. Given a flat connection $\alpha$ on $E$, it is clear that
the subspace $C^\bullet(K,L;E)=\ker(p)$ is also a subcomplex with the
differential $d_\A$. We have an exact sequence:
\[
0\rightarrow C^\bullet(K,L;E)\rightarrow C^\bullet(K;E)\stackrel{p}{\rightarrow} C^\bullet(L;E)\rightarrow 0
\]
This exact sequence naturally splits because the spaces come with bases enumerated by simplices. Denote
this splitting by $j: C^\bullet(L;E) \rightarrow C^\bullet(K;E)$ .

The scalar product $\langle \cdot , \cdot \rangle_K$ on $C^\bullet(K;E)$
defines the natural scalar product on the subcomplex $C^\bullet (K,L;E)$.
The splitting defines the scalar product on $C^\bullet(L;E)$.
\[
\langle \phi, \psi \rangle_{L}=\langle j(\phi), j(\psi) \rangle_{K}
\]
Define the bilinear form $\langle \phi, \psi \rangle_{K,L}$ on $\phi, \psi \in C^\bullet (K; E)$
as
\[
\langle \phi, \psi \rangle_{K, L}=:\langle \phi, \psi \rangle_{K} -
\langle p(\phi), p(\psi) \rangle_{L}.
\]
when restricted to $C^\bullet (K, L; E)$ it coincides with the natural scalar product
on this space induced from scalar product on cochains on $K$.

There is a natural analog of the
Green formula for $\phi \in C^q(K,E)$ and $\psi \in C^{q+1}(K,E)$:
\[
\langle d_\A\phi, \psi \rangle_{K,L}=\langle \phi, d_\A^\vee\psi \rangle_{K,L}+
\langle p(\phi), \psi_{\textup{norm}} \rangle_{L}.
\]
Here
\[
\langle p(\phi), \psi_{\textup{norm}} \rangle_{L}=\langle p(\phi), p(d_\A^\vee\psi)\rangle_L-\langle p(d_\A\phi), p(\psi)\rangle_L
\]
We will call $\psi_{\textup{norm}}$ the normal component of $\psi$ at the boundary $L$.
This defines the combinatorial Laplacian $\Delta_{K,L}$ on cochains $C^\bullet(K,L;E)$,
as above in \eqref{lapl}. The operator $\Delta_{K,L}$ corresponds to the analytic Hodge Laplacian with relative
boundary conditions. In degree zero, $\Delta_{K,L}$ is simply the restriction of $\Delta_K$
to $C^0(K,L;E)\subset C^0(K;E)$.

\section{The approximation theory of Dodziuk}\label{Dodziuk}

Here we will outline the relation between combinatorial constructions
and the smooth theory.

\subsection{Simplicial approximation of vector bundles}
Here we will focus the question how simplicial complexes with vector bundles
and connections on them can be induced by a  smooth triangulation on a manifold and a
vector bundle on it with a connection. Consider
an oriented compact smooth $M$ and a
vector bundle $E$, with a connection $\A$.

From now on let $K$ be a simplicial complex of a smooth triangulation of $M$,
identified with its embedding in $M$.

Denote by $D(K)$ the double of the simplicial complex $K$. Fix an embedding $D(K)$ in $M$ which makes it a
simplicial decomposition of $M$.  For each $\sigma \in K$ denote
by $p_\sigma$ the point in the interior of $\sigma$ at which $\sigma$
intersects with its dual, recall Figure 1. Define the combinatorial vector bundle $E^c$, with a connection
on it as
\begin{align*}
E^c_\sigma = E_{p_\sigma}, \quad \A^c(\sigma, \tau) = \A(p_\sigma, p_\tau),
\tau \in \partial \sigma.
\end{align*}
where $\A(p_\sigma, p_\tau)$ is the holonomy along the edge of $D(K)$
connecting $p_\sigma$ and $p_\tau$. This is the only place where the actual
choice of centers is relevant for the construction.

When the connection is flat, then in the notation of Definition \ref{flatness} the vertices 
$(p_\eta, p_{\tau^\pm}, p_\sigma)$ form a contractable loop, so that the holonomy
along that loop vanishes and the combinatorial connection $\A^c$ is flat by definition. 
We denote by $(C^\bullet(K,E^c),d_c)$ the associated cochain complex.

Assume that the combinatorial connection is flat and 
recall the definition of the complex $(C^\bullet(K,F_\A),d_F)$ which has been considered in
\cite{Dod:FDA} and \cite{Mue:ATA} in the context of Dodziuk's
approximation theory. Here, $F_\A$ denotes the local system of flat sections.

For a simplex $\sigma\in K$ the \emph{open star}  of $\sigma \in K$, $\textup{St}(\sigma)$,
is defined as the union of $\sigma$ and of interiors of all simplices $\tau \in K$,
which contain $\sigma$ as part of their boundary: $\sigma \in \partial \tau$.
Every vector $v\in E_\sigma$ in the fiber $E_\sigma$ over $p_\sigma$ extends by the parallel
transport $\A$ uniquely to a flat section $\Phi_\A v $ of $E$ over the open star of $\sigma \in K$.

Consequently we can identify each $\phi \in C^q(K,E^c)$ with a
mapping $\Phi_\A \phi$ that assigns to any $\sigma \in K_q$
a flat section $\Phi_\A \phi(\sigma) \in F_\A$. Under such an identification
we find for any $v\otimes \sigma^* \in C^q(K,E^c)$, where $\sigma \in K_q, v\in E_\sigma$
and $\sigma^*$ is the the dual to $\sigma$, that
\begin{align*}
d_c^\Phi \left(\Phi_\A (v\otimes \sigma^*) \right) :=
\Phi_\A \circ d_c \circ \Phi_\A^{-1} \left(\Phi_\A (v\otimes \sigma^*) \right)
= \left( \Phi_\A v \right) \otimes d\sigma^*.
\end{align*}

Thus the identification $\Phi_\A: (C^*(K,E^c),d_c) \to (C^*(K,F_\A),d_c^\Phi)$
is an isomorphism of complexes. We have studied $(C^*(K,E^c),d_c)$
in the previous section only as it bears the advantage of being defined without using analytic data.

\subsection{Simplicial approximation for Laplacians}

\subsubsection{The Whitney map.}\label{WM}

Let $K$ be a smooth CW complex on $M$. Recall the definition of the Whitney map
$$
W: C^\bullet(K,F_\A) \to \Omega^\bullet(M,E),
$$
where $C^\bullet(K,F_\A)$ is defined above. The Whitney map was introduced by Whitney
\cite{Whi:GIT}, see also its use \cite{Dod:FDA, Mue:ATA} in the
spectral approximation theory.

Let $K_0 = \{\sigma_0,..,\sigma_d\}$ be the vertices of $K$.
Every $\sigma_i$ defines a barycentric coordinate function $\mu_i$ on $M$.
Each $\mu_i$ is a continuous function on $M$ with support in the
open star of $\sigma_i$. Moreover, since the triangulation is smooth,
the restriction of $\mu_i$ to any simplex is smooth. Hence, the
differential $d\mu_i\in L^2\Omega^1(M)$ exists in the distributional sense.
Given an $q$-simplex $\tau = [\sigma_{i_0},..,\sigma_{i_q}] \in K_q$,
with ascending sequence $\{i_k\}_k$ of indices, and
some flat section $v\in F_\A$ supported over the open star of $\tau$,
we define the Whitney map by\footnote{Whitney map is a particular choice of a linear spline.}
\begin{align*}
W(v\otimes \tau^*) &= v \otimes \left(q! \sum_{k=0}^q (-1)^k \mu_{i_k}
d\mu_{i_0} \wedge \cdots \widehat{\mu_{i_k}} \cdots \wedge d\mu_{i_q}\right), \ q>0, \\
W(v \otimes \sigma_{i_0}^*) &= v \otimes \mu_{i_0}, \ q=0.
\end{align*}

\subsubsection{The metric structures.}\label{whitney-subsubsection}
The Riemannian structure $g$ on $M$ and the Hermitian structure on $E$ define an $L^2$-structure on $E$-valued
differential forms $\Omega^\bullet(M,E)$. The corresponding
completion is usually denoted by $L^2\Omega^\bullet(M,E)$. Define the Hermitian
metric on the combinatorial vector bundle $E^c$ over $K$ by taking the Hermitian product of
sections over a center $p_\sigma$ of each simplex $\sigma \in K$. We define the
scalar product of any $\phi, \psi \in C^\bullet (K,E^c)$ in terms of the Whitney map as
\begin{align}
\langle \phi, \psi \rangle^W_K := \langle W \Phi_\A \phi, W \Phi_\A \psi \rangle_{L^2}.
\end{align}
Note that for the Whitney scalar product $|(\sigma, \tau)|\neq 0$ only if there
exists a simplex $\eta\in K$ such that $\sigma, \tau \subset \pa \eta$. This is an example of a \emph{local}
scalar product (not to be confused with a diagonal scalar products).

This definition provides an explicit example of our general construction
in \eqref{inner}. Indeed, consider $\phi = v\otimes \sigma^*$ and
$\psi=w\otimes \tau^*$, where $\sigma, \tau \in K_q$ and $v\in E_\sigma, w\in E_\tau$.
Then we have
\begin{align*}
\langle\phi, \psi \rangle^W_K &= 
\langle W \Phi_\A (v\otimes \sigma^*), W \Phi_\A (w\otimes \tau^*)\rangle_{L^2}
= \int_M h(\Phi_\A v, \Phi_\A w) \langle W\sigma^*, W\tau^*\rangle_g \, \textup{dvol}_g \\
&= h_{p_\sigma}(v, \A(p_\sigma, p_\tau) w) \langle W\sigma^*, W\tau^*\rangle_{L^2} 
=: h^c_{p_\sigma}(v, \A(p_\sigma, p_\tau) w)  |(p_\sigma, p_\tau)|.
\end{align*}

\subsubsection{ The approximation theory.}
Let $E$ be a flat Hermitian vector bundle.
The de Rham map $R$ associates to any $f\in \Omega^q(M,E)$
a cochain $Rf \in C^q(K,E^c)$ as follows. In a local neighborhood
of any simplex $\sigma \in K_q$ we fix a basis $\{v_1,..,v_p\}, p=\textup{rank}E,$
of flat sections of $E$ and write $f$ locally as
$f= \sum_{i=1}^p v_i \otimes f_i$, where each $f_i$ is an ordinary differential form.
Then
\begin{align}\label{de-Rham}
Rf(\sigma) := \sum_{i=1}^p v_i(p_\sigma) \int_\sigma f_i.
\end{align}
where $p_\sigma$ is the center of $\sigma$.

The approximation theory of forms by cochains is based on the fact that the composition
of the de Rham and Whitney maps converge to identity when the mesh of
the triangulation $K$ goes to zero under standard subdivisions,
see \cite[Theorem 3.7]{Dod:FDA} and also \cite{Mue:ATA} for twisted setup.
The crucial results is the following
\begin{thm}\cite[Th. 5.30]{Dod:FDA}, \cite[Th. 4.6]{Mue:ATA}
The combinatorial zeta functions $\zeta (s, \Delta^{K(\eta)}_\A)$
converge under standard subdivisions uniformly on compact subsets of
the complex half plane $\textup{Re}(s)>\dim M/2$ to the analytic zeta function
$\zeta(s, \Delta_{\textup{abs}})$, as the mesh $\eta$ of the triangulation $K(\eta)$ goes to zero.
\end{thm}

Similar result holds for combinatorial zeta functions of $\Delta_{K,L}$
which converge to analytic zeta functions of the Laplacians with relative
boundary conditions.

Finally, note that defining the combinatorial Laplacian $\Delta_K^q$ from an inner product on $C^*(K,E^c)$
depends on the choice of centers in $K^\vee$. However, $\Delta_K^q$ on $C^q(K,E^c)$ is
equivalent to the corresponding Laplacian on $C^q(K,F_\A)$ under the isometric identification
$\Phi_\A$, in particular combinatorial Laplacians defined with respect to different choices
of centers, are unitarily or orthogonally equivalent and have the same spectral properties.

\section{Classical field theory for free scalar fields}\label{DN}

Here we will outline classical field theory for free scalar field on
a metrized simplicial complex with boundary. Having in mind corresponding
Gaussian quantum field theory we will call such classical field
theory Gaussian.

\subsection{Classical scalar free theory on Riemannian manifolds}

\subsubsection{The action and its minimizers.} Let $(M,g)$ be a Riemannian compact oriented manifold,
possibly with boundary and with a Hermitian vector bundle $(E,h)$ over $M$ with a Hermitian
connection $\A$ on it. For classical scalar field theory on $M$
with values in $E$, the space of fields is the space of sections of $E$. We assume fields are smooth sections of $E$.
The action functional is
\begin{equation}\label{sc-field-cont}
S_{M}(\phi) := \frac{1}{2}\langle d_\A   \phi, d_\A \phi \rangle_M +
\frac{m^2}{2}\langle \phi, \phi\rangle_M.
\end{equation}
The constant $m^2$ has the meaning of the square of the mass of a particle.

For the variation of the action we have
\begin{align*}
\delta S_{M}(\phi) = \langle(d_\A^* d_\A  + m^2) \phi, \delta \phi \rangle_M +
\langle \partial^\A_\nu \phi, \delta \phi\rangle_{\pa M}.
\end{align*}
Here, $\partial^\A_\nu = \iota^* (\iota_\nu d_\A \psi)$ is the covariant normal derivative,
with $\iota: \pa M \hookrightarrow M$ denoting the natural inclusion, and $\iota_\nu$ the contraction with the unit
normal vector field on $\pa M$.
The boundary term vanishes for Dirichlet boundary conditions, i.e. if we assume $\phi|_{\partial M} =\eta$.
Euler Lagrange equations for this action are
\begin{align}\label{EL}
(d_\A^* d_\A  + m^2) \phi = 0,
\end{align}
which for fixed boundary conditions $\phi|_{\partial M} =\eta$ admit
a unique solution given in terms of the corresponding Poisson map $\phi_c=P_M\eta$. The value of
$S_{M}$ at the critical point $\phi_\eta=P_M\eta$ is then
\begin{align}\label{critical-M}
S_{M}(\phi_\eta) = \frac{1}{2}\langle R^M \eta, \eta \rangle_{\partial M},
\end{align}
where $R^M = \partial^\A_\nu \circ P_M$ is the Dirichlet to Neumann operator for $\Delta_M+m^2$ given
in terms of the Poisson operator $P_M$.

The action \ref{sc-field-cont} is local, i.e. for any open submanifold $N\subset M$ it satisfies the following property:
\[
S_{M}(\phi)= S_{M\backslash N}(\phi|_{M\backslash N}) + S_N(\phi|_N).
\]
When $N$ is of codimension one, disjoint from $\partial M$, the second term on the right side is absent and the
critical value of the action has the following gluing property:
\[
S_M(\phi_{\eta'}) = \min_{\eta }(S_{cl(M\backslash N)}(\phi_{\eta, \eta, \eta'})),
\]
where $cl(M\backslash N)$ denotes the closure of $M\backslash N$ with boundary comprised
of three components $N\sqcup N \sqcup \partial M$ and $\eta'$ is the fixed boundary value of fields on $\partial M$.

\subsection{Classical scalar free theory on metrized simplicial complexes}

In this section and in the rest of the paper we will focus on the scalar Bose field
when fields are elements of $C^0(K;E)$. Theories with fields from $C^1(K;E)$
and from higher degree cochains usually involve gauge symmetry and we will not consider them here.
In this section we assume that Riemannian norms on simplicial complexes are local.

\subsubsection{The action}

Let $(K,|\cdot |)$ be a metrized simplicial complex of dimension $n$ with $(n-1)$-dimensional boundary subcomplex $L$.
Let $(E,h)$ a combinatorial Hermitian vector bundle over $K$, with a Hermitian connection $\A$.

Define the action\footnote{Note that though we used orientation on $K$ for definition of $d_\A$,
orientability is not required for the definition of $\langle d_\A   \phi, d_\A \phi \rangle_{K}$ and the action
makes sense on non-orientable complexes as well.} of the free Bose scalar field theory as the following function on $C^0(K,E)$ as
\begin{equation}\label{gauss-act}
\begin{split}
S_K(\phi) &:= \frac{1}{2}\langle d_\A   \phi, d_\A \phi \rangle_{K} +
\frac{m^2}{2}\langle \phi, \phi \rangle_{K}  \\
&= \frac{1}{2} \sum_{\tau, \tau' \in K} h(d_\A \phi(\tau), \A(\tau, \tau') d_\A \phi(\tau')) |(\tau, \tau')| \\
&+ \frac{m^2}{2} \sum_{\sigma, \sigma' \in K} h(\phi(\sigma), \A(\sigma, \sigma') m^2(\sigma')\phi(\sigma')) |(\sigma, \sigma')|.
\end{split}
\end{equation}
When $L\subset K$ is a boundary subcomplex, define  the boundary action as
\begin{equation}\label{gauss-L}
\begin{split}
S_L(\phi) &:= \frac{1}{2}\langle d_\A   \phi, d_\A \phi \rangle_{L} +
\frac{1}{2}\langle \phi, m^2 \phi \rangle_{L}  \\
&= \frac{1}{2} \sum_{\tau, \tau' \in L} h(d_\A \phi(\tau), \A(\tau, \tau') d_\A \phi(\tau')) |(\tau, \tau')| \\
&+ \frac{m^2}{2} \sum_{\sigma, \sigma' \in L} h(\phi(\sigma), \A(\sigma, \sigma')\phi(\sigma')) |(\sigma, \sigma')|,
\end{split}
\end{equation}
Here $\phi \in C^0(L,E)$ and we emphasize that $S_L$ is defined with respect to the extrinsic\footnote{An intrinsic metrization of the
boundary $L$ never appears in our discussion.} metrization of $L$,
obtained as the restriction of the Riemannian norm on $K$ to $L$. In other words, weights in $S_L$ are the same as in $S_K$. Geometrically, in case of the Whitney scalar
product, the action depends on the Riemannian norm in a star neighborhood of $L$. Define
\begin{equation}\label{S-K-L}
S_{K,L}(\phi) := S_K(\phi) - S_L(\phi|_L).
\end{equation}

The action $S_{K,L}(\phi)$ by definition comprises the self-interaction between interior chains in $K\backslash L$
as well as the interaction between interior and the boundary chains, however does not contain interactions between the boundary chains themselves, i.e. it does not have terms $h(\phi(\sigma), \A(\sigma,\tau)\phi(\tau))$ where
both $\sigma$ and $\tau$ are in $L$.

For local scalar products on cochains the action is also local in the following sense.
Assume that $K$ admits a subcomplex $L$ that consists of three connected components
$L= L_1 \sqcup L_2 \sqcup L_3$. Let $L_1$ and $L_2$ be isometrically identified via $f$.
This defines a new chain complex $K_f$ with a single boundary subcomplex $L_3$ and a metrization induced from $K$.
Let $\phi$ be a cochain on $K$, such that $\phi|_{L_1}=f^*\phi|_{L_2}$.
Then $\phi$ lifts to well-defined cochain $\phi_f$ on $K_f$ and the locality of the action is expressed as follows
\[
S_K(\phi) = S_{K_f}(\phi_f)+S_{L_2}(\phi|_{L_2}).
\]

\subsubsection{Minimizers of the Gaussian action with Dirichlet boundary conditions.}
Let us describe the minimizer of the Gaussian action $S_{K,L}$ over the space
of cochains $\phi \in C^0(K,E)$ with given Dirichlet boundary condition $\phi|_L =\eta \in C^0(L,E)$.
Because for Dirichlet boundary conditions the variation $\delta \phi$ vanishes at $L$,
the variations of $S_K, S_{K,L}$ are identical and
\begin{align}\label{var-s}
\delta S_{K,L}(\phi) = \delta S_K(\phi) = \langle (\Delta_K + m^2) \phi, \delta \phi \rangle_K =0.
\end{align}
Explicitly, we have
\begin{align*}
\delta S_K(\phi)
&= \sum_{\tau, \tau' \in K} h(d_\A \phi(\tau), \A(\tau, \tau') d_\A \delta \phi(\tau')) |(\tau, \tau')| \\
&+ m^2\sum_{\sigma, \sigma' \in K} h(\phi(\sigma), \A(\sigma, \sigma')\delta\phi(\sigma')) |(\sigma, \sigma')|.
\end{align*}
For each vertex $\sigma \in K_0$, consider sets
\begin{equation}
\begin{split}
&\mathscr{V}(\sigma) := \{(\tau,\tau',\sigma') \mid \sigma \in \pa \tau, |(\tau,\tau')|\neq 0, \sigma' \in \pa \tau'\}, \\
&\mathscr{U}(\sigma) := \{\sigma' \mid \exists (\tau,\tau'): (\tau,\tau',\sigma') \in \mathscr{V}(\sigma)\}.
\end{split}
\end{equation}
In case of a local
(Whitney) inner product on a simplicial complex , $\mathscr{U}(\sigma) = \overline{\textup{St}(\sigma)} \cap K_0$ consists of
all vertices in the closure of the \emph{open star} of $\sigma$. We refer to $\mathscr{U}(\sigma)$ as the local
neighborhood of $\sigma\in K_0$. Then
\begin{align*}
\delta S_K(\phi)
&= \sum_{\sigma\in K} \sum_{\mathscr{V}(\sigma)} (-1)^{(\tau,\sigma) + (\tau',\sigma')}
h(\phi(\sigma), \A(\sigma,\sigma')\delta\phi(\sigma')) |(\tau,\tau')| \\
&+ m^2\sum_{\sigma\in K} \sum_{\sigma'\in K} h(\phi(\sigma), \A(\sigma,\sigma')\delta \phi(\sigma')) |(\sigma,\sigma')|.
\end{align*}
Consequently, the Euler-Lagrange equations for \eqref{var-s} can be written as
\begin{align*}
& \sum_{\mathscr{V}(\sigma)} (-1)^{(\tau,\sigma) + (\tau',\sigma')}
h(\phi(\sigma'), \A(\sigma',\sigma)\phi(\sigma)) |(\tau,\tau')| \\
&+ m^2\sum_{\sigma'\in K} h(\phi(\sigma'), \A(\sigma',\sigma) \phi(\sigma')) |(\sigma,\sigma')| = 0.
\end{align*}
We have such equation for each $\sigma\in (K\backslash L)_0$. Each equation is a linear
equation involving  vertices in the local neighborhood $\mathscr{U}(\sigma)$.

Because of convexity\footnote{i.e. for all $0<\theta<1$ and
$\phi, \psi \in C^0(K,E)$ with $\phi|_L, \psi|_L=\eta$
$$S_{K, L}(\theta\phi+(1-\theta)\psi)\leq \theta S_{K, L}(\phi)+(1-\theta)S_{K, L}(\psi).$$
} of $S_{K,L}$ on fibers of $p$, the solution to this difference equation
with Dirichlet boundary conditions $\phi|_L=\eta$, exists and is unique.
We denote it by $\phi_\eta$. The solution $\phi_\eta$ is linear in $\eta$ and hence we can
define the discrete version of the Poisson kernel as
\[
\phi_\eta(\sigma) = \sum_{\sigma'\in L} P_{K,L}(\sigma, \sigma') \eta(\sigma'),
\]
or $\phi_\eta =P_{K,L}\eta$. The value of $S_{K,L}(\phi_\eta)$ at the critical
point $\phi_\eta$ is quadratic in $\eta$
and we can write
\begin{align}\label{DN-def}
S_{K,L}(\phi_\eta) = \frac{1}{2} \langle \eta, R^K_L \eta \rangle_L,
\end{align}
where $R^K_L$ is the discrete version of the Dirichlet-to-Neumann operator.

\begin{remark} For a diagonal inner product, the Poisson map and hence also the Dirichlet-to-Neumann operator
are explicit (see \cite[Theorem 2.1]{ClaMaz:CSA}). In this case
the boundary value problem
\begin{equation}
\label{BVP}
\begin{split}
(\Delta_K + m^2) \phi &= 0, \textup{on vertices of} \ K\backslash L, \\
\phi (\sigma) &=\eta (\sigma), \textup{for all} \ \sigma \in L.
\end{split}
\end{equation}
has a unique solution $\phi=P_K \eta \in C^0(K, E)$, and assuming
the scalar product on $C^\bullet (K,E)$ is diagonal,
the Poisson operator $P_K$ can be described explicitly by
\begin{equation}
 \begin{split}
 P (P_K \eta)  = - \left(P \Delta_K P + m^2\right)^{-1} P \Delta_K \eta, \quad
(\textup{Id}-P) R_K = \eta.
 \end{split}
\end{equation}
Here $P$ is the natural projection $P: C^0(K, E) \to C^0(K, L; E) \subset C^0(K, E)$
acting trivially on fibres of $E$.
\end{remark}

Values of the action functional on critical points have the following
gluing property in case of local scalar products.
Assume as before that $K$ admits a subcomplex $L$ that consists of three connected components
$L= L_1 \sqcup L_2 \sqcup L_3$. Let $L_1$ and $L_2$ be isometrically identified via $f$.
This defines a new chain complex $K_f$ with a single boundary subcomplex $L_3$ and a metrization induced from $K$. Then
\begin{equation}\label{gl-cr-act}
S_{K_f, L_3}(\phi_{\eta_3}) = \min_{\eta_{12}}(S_{K, L}(\phi_{\eta_{12}, \eta_{12}, \eta_3})+S_{L}(\eta_{12})).
\end{equation}

\section{Discrete Quantum Gaussian field theory}\label{QFT}

\subsection{The partition function.}

\subsubsection{Category of cobordisms}
Recall that  a quantum field theory on $n$-dimensional Riemannian manifolds can be regarded
as functor from the category of  Riemannian $n$-dimensional cobordisms to the category of vector spaces.
An object in the category of Riemannian $n$-dimensional cobordisms is a smooth, oriented, compact $(n-1)$ dimensional
Riemannian manifold $N$ with an $n$-dimensional smooth, oriented Riemannian collar
$\mathscr{U}(N)\cong (-\epsilon, \epsilon)\times N$ equipped with a product Riemannian metric.

A morphism between two objects $N_1$ and $N_2$ is a smooth, oriented, compact $n$-dimensional
manifold $M$ such that $\pa M= \overline{N_1}\sqcup N_2$ with collars at each connected component of the boundary.
The composition of morphism is gluing such that collars on the common boundary agree
with both morphisms.

The combinatorial analog of this category is the category of metrized\footnote{Recall that a metrized simplicial
complex is a simplicial complex $L$ with scalar product on the corresponding cochain complex $C^\bullet(L)$
(e.g. over $\R$). An isometry of simplicial complexes is an isomorphism of simplicial complexes which yields an
isometry between the corresponding scalar product cochain spaces.} local $n$-dimensional simplicial complexes.
An object in this category is an $(n-1)$-dimensional simplicial complex $L$ with a metrized $n$-dimensional collar complex $\mathscr{U}(L)$ which is
homotopy equivalent to $L$. The metrization on $L$ is required to arise from the metric structures on $\mathscr{U}(L)$ by restriction.

For local metric structures, $\mathscr{U}(L)$
is the star neighborhood of $L$. For example, this is the case for Whitney scalar product.

\subsubsection{Framework of QFT}\label{framework} The framework of an $n$-dimensional local quantum field theory,
cf. \cite{At} and \cite{Se}, applied to metrized simplicial complexes consists of the following two assignments:

\begin{itemize}
\item To each $(n-1)$-dimensional metrized oriented simplicial complex $L$ with an $n$-dimensional collar $\mathscr{U}(L)$ we assign a
$\C$-vector space $H(L)$ with  a non-degenerate linear pairing
\[
( \ , \, )_{L}: H(L')\otimes H(L)\to \CC,
\]
where the simplicial complexes $L'$ and $L$ differ only by orientation.

\item An orientation reversing automorphism $\sigma_L$ of $L$ lifts to
a an isomorphism of vector spaces $\widehat{\sigma}_L: H(L)\to H(L')$, inducing the structure
of a Hilbert space on $H(L)$\footnote{In case of non-orientable simplicial complexes, we only assign a Hilbert space structure on $H(L)$.}. An orientation preserving isometry
$f: L_1\to L_2$ of metrized simplicial complexes lifts to an
isometry $\widehat{f}: H(L_1)\to H(L_2)$.

\item The collar $\mathscr{U}(L)$ is separated by $L$ into subcomplexes $L_{\pm}$, such that
$\mathscr{U}(L)=L^+\cup_L L^-$. We will call $L^+$ and $L^-$ right and left neighborhoods of $L$
respectively. To each metrized $n$-dimensional simplicial complex $K$ such that $L^+\subset K$ is
a metrized simplicial subcomplex we assign the vector
\[
Z_{K,L}\in H(L).
\]
\end{itemize}

These data should satisfy certain axioms. One of the most important axioms
is the locality property of the
partition function, also known as the gluing axiom. Assume
that the boundary of a simplicial complex $K$ has three connected components $L_1$, $L_2$ and $L_3$.
Assume that the corresponding right neighborhoods $L^+_1, L^+_2, L^+_3$ are disjoint.
We write $L:= L_1 \sqcup L_2 \sqcup L_3$. Assume that metrized simplicial complexes
$L_1$ and $L_2$ together with their collars are isometric via an
isometry $f$ of simplicial complexes. Let $K_f$ be the result of the gluing the boundary component $L_2$ to
$L_1$ via the orientation reversing isometry $\sigma_{L_2} \circ f$, and that the scalar product
on $K$ is the pullback of the scalar product on $K_f$.

The partition function $Z_{K,L}$ is a vector in $H(L_1)\otimes H(L_2)\otimes H(L_3)$.
Then we require the following
\begin{equation}\label{glu}
(( \ , \, )_{L_2}\otimes \textup{id})
(\widehat{\sigma_{L_2} \circ f}\otimes \textup{id} \otimes \textup{id})
(Z_{K,L} )=Z_{K_f, L_3}\in H(L_3).
\end{equation}

\subsubsection{Gaussian QFT} Now let us construct a Gaussian
quantum field theory which satisfies all
these properties.

\begin{itemize}

\item To an  $(n-1)$ dimensional metrized simplicial complex $L$ with a local scalar product\footnote{in the
sense of \S \ref{locality}.} and an $n$-dimensional collar $\mathscr{U}(L)$ and a Hermitian vector bundle $(E,h)$
we assign the Hilbert space
\[
H(L)=L^2(C^0(L;E)),
\]
with the scalar product
\[
( f, g )_{H(L)} :=\int_{C^0(L;E)}  \overline{f(\eta)}g(\eta)\, e^{-S_L(\eta)} d\eta.
\]
Here $S_L$ is the classical Gaussian action defined as in \eqref{gauss-L}
defined with respect to the scalar product on cochains given by restriction of the metric structure on the collar $\mathscr{U}(L)$. We integrate with respect to the Euclidean measure $d\eta$ on $C^0(L;E)$. Note that $S_L(\eta)$ is positive definite.

\item If we assume that $K$ is oriented, then the pairing $( \ , \, )_{L}: H(L')\otimes H(L)\to \CC$
is given by the composition of the natural $L^2$-scalar product on $H(L)$
and the lift $\widehat{\sigma}_L$ of the orientation reversing automorphism
$\sigma_L: L\to L'$. Gaussian action is invariant under change of orientation, i.e.
$S_{\sigma_L(L)}=S_L$ and the \emph{Hermitian axiom} of Atiyah-Segal is trivially satisfied.
Moreover, the pullback $\sigma_L^*$ to cochains is trivial in degree zero, so that
$H(L)=H(L')$ and the pairing $( \ , \, )_{L}$ coincides with the scalar product on $H(L)$.

\item To the metrized simplicial complex $K$ with boundary $L\subset K$, equipped with its right neighborhood
$\mathscr{U}(L)^+\subset K$ fitting $K$ and with the weights on $\mathscr{U}(L)^+$ given by weights on $K$,
we will assign the \emph{partition} function
$Z_{K, L} \in H(L)$ as follows. Consider the natural projection $p: C^0(K;E)\to C^0(L;E)$
which is the restriction to the boundary. For each $\eta \in C^0(L;E)$ we set
(recall \eqref{S-K-L})
\begin{equation}\label{pf}
Z_{K, L}(\eta) := \int_{p^{-1}(\{\eta\})} \exp \left( - S_{K,L}(\phi) \right) \, d\phi,
\end{equation}
\end{itemize}
Here the integration measure is the Euclidean measure on the vector space corresponding to the affine space $p^{-1}(\{\eta\})\subset C^0(K;E)$.
The integral is convergent, so partition function $Z_{K,L}(\eta)$ is defined because the classical action $S_K$ is
positive and strictly convex on each fiber $p^{-1}(\eta)$. We have $Z_{K,L}\in H(L)$ with
$\|Z_{K,L}\|^2_{H(L)} = Z_{\widetilde{K}, \varnothing}\in \C$\footnote{We identify $H(\varnothing)\cong \C$.}, where
$\widetilde{K}:=K\cup_L \sigma_L^{-1}(K)$ is the closed double of $K$. \medskip

The gluing property is an exercise on Fubini's theorem.

\begin{theorem}\label{gluing-theorem}
The partition function $Z_{K,L}$ satisfies the gluing axiom.
\end{theorem}

\begin{proof}
Recall that in \eqref{glu} $K$ is a simplicial complex with the boundary $L=L_1\sqcup L_2\sqcup L_3$
and $L_1$ and $L'_2$ are isometric
via orientation preserving isometry $f_\sigma = \sigma_{L_2} \circ f: L_1\to L'_2$ of simplicial complexes.  Let us compute the left side of
the identity \eqref{glu}.  The pullback of $f_\sigma$ to cochains defines
$f_\sigma^*: C^0(L'_2;E) \to C^0(L_1;E)$. The mapping $\widehat{f_\sigma}: H(L_1)\to H(L_2')$ is the pull-back of $f^*_\sigma$.
For $(\eta'_2,\eta_2,\eta_3) \in C^0(L'_2;E) \otimes C^0(L_2;E) \otimes C^0(L_3;E)$ we have
\[
((\widehat{f_\sigma}\otimes \textup{id} \otimes \textup{id})
Z_{K, L})(\eta'_2,\eta_2,\eta_3)= Z_{K, L}(f_\sigma^*\eta'_2,\eta_2,\eta_3).
\]
Let $K_f$ be the simplicial complex obtained from $K$ by gluing the boundary component $L_2$ to $L_1$ via
the orientation reversing isomorphism $f_\sigma$. The scalar product on $K_f$ is the
pullback of the scalar product on $K$. Then any $\phi\in  C^0(K_f;E)$ with
$\phi|_{L_2}=\eta_2$ and $\phi|_{L_3}=\eta_3$ can be regarded as $\phi\in C^0(K;E)$ with
$\phi|_L = (f^* \eta_2,\eta_2,\eta_3)$. Moreover, by construction
\[
S_{K}(\phi) = S_{K_f}(\phi) + S_{L_2}(\eta_2).
\]
Consequently we obtain
\begin{equation}\label{ZK}
\begin{split}
&(( \ , \, )_{L_2}\otimes \textup{id})
(\widehat{f_\sigma}\otimes \textup{id} \otimes \textup{id})
(Z_{K,L} ) (\eta_3) \\ &= \int_{C^0(L_2;E)}  Z_{K, L}(f^*
\eta_2,\eta_2,\eta_3) \, e^{-S_{L_2}(\eta_2)} d\eta_2
=Z_{K_f, L_3}(\eta_3).
\end{split}
\end{equation}
\end{proof}

Write $\langle \cdot, \cdot \rangle_0$
for the Euclidean inner product on cochains (in the simplex basis). Then there exists an endomorphism $Q_K$ on $C^0(K,E)$, Hermitian with respect to the Euclidean inner product, such that
$\langle \phi, \psi \rangle_K = \langle Q_K \phi, \psi \rangle_0$ for any $\phi, \psi \in C^0(K,E)$.
We write $\Delta^{loc}_{K} = Q_{K\backslash L} \circ \Delta_{K}$\footnote{The operator $\Delta^{loc}_K$
is local in a sense that it acts between nearest neighboring vertices.}, where $\Delta_{K}$ is the combinatorial
Laplacian on $K$ with Dirichlet boundary conditions at the boundary subcomplex $L$. 
The Gaussian integral (\ref{pf}) is then easy to compute.
\[
Z_{K,L}(\eta)=(2\pi)^{\frac{|K\backslash L|}{2}}\frac{e^{-S_{K,L}(\phi_\eta)}}
{\sqrt{\det(\Delta^{loc}_{K}+m^2Q_{K\backslash L})}}.
\]
If we substitute this formula into the gluing identity \eqref{ZK} for the partition function
we arrive at the following relation
\begin{align*}
Z_{K_f, L_3}(\eta_3) 
&= (2\pi)^{\frac{|K_f\backslash L_3|}{2}} \frac{e^{-S_{K_f,L_3}(\phi_{\eta_3})}}
{\sqrt{\det(\Delta^{loc}_{K_f}+m^2Q_{K_f\backslash L_3})}}
\\ &= (2\pi)^{\frac{|K\backslash L|}{2}}\int\limits_{C^0(L_2;E)}  \frac{e^{-S_{K,L}
(\phi_{f^*\eta_2,\eta_2,\eta_3})-S_{L_2}(\eta_2)}}
{\sqrt{\det(\Delta^{loc}_{K}+m^2Q_{K\backslash L})}} d\eta_2
\end{align*}

From there we get the gluing identity for critical values of the classical action (\ref{gl-cr-act}) in the exponent
and, in view of \eqref{DN-def} the following identity of determinants in the pre-exponent:
\begin{align}\label{gluing-loc}
\frac{\det(\Delta^{loc}_{K_f}+m^2Q_{K_f \backslash L_3})}{\det(\Delta^{loc}_{K}+m^2Q_{K\backslash L})}=\det \mathscr{R}^{loc}_c(K_f,L_2),
\end{align}
where $\mathscr{R}^{loc}_c(K_f,L_2)= Q_{L_2} \circ \mathscr{R}_c(K_f,L_2)$ and $\mathscr{R}_c(K_f,L_2)$ is defined by
\begin{align}\label{DN-gluing}
\langle \mathscr{R}_c(K_f,L_2) \eta_2, \eta_2 \rangle_{L_2} := \langle R^K_L (\eta_2, \eta_2, 0), (\eta_2, \eta_2, 0)\rangle_{L_1\sqcup L_2} + \langle \Delta_{L_2} \eta_2, \eta_2 \rangle_{L_2}.
\end{align}
This proves Theorem \ref{main1}. 

It remains to identify a gluing relation for the endomorphism $Q_K$, associated to the quadratic form
of the scalar product on $C^0(K,E)$. With respect to the direct sum decomposition $C^0(K_f \backslash L_3,E) =
C^0(K\backslash L, E) \oplus C^0(L_2,E)$ we may write in the basis defined by the duals of the vertex elements
\begin{equation}\label{block}
Q_{K_f\backslash L_3} = \left(
\begin{array}{cc}
Q_{K\backslash L} & A \\
A^t & Q_{L_2}
\end{array} \right),
\end{equation}
where $A:C^0(L_2,E) \to C^0(K\backslash L,E)$ describes the interaction in the inner product of $K$
between vertices at the boundary subcomplex $L_2$ with the interior vertices. From the block representation \eqref{block}
and from \eqref{gluing-loc} we find
\begin{equation}\label{gluing-loc-formula}
\begin{split}
\frac{\det(\Delta_{K_f}+m^2)}{\det(\Delta_{K}+m^2)} &=
\frac{\det Q_{K \backslash L}\det Q_{L_2}}{\det Q_{K_f \backslash L_3}}\det \mathscr{R}_c(K_f,L_2) \\
&=\frac{\det Q_{L_2}}{\det (Q_{L_2} - A^t \circ Q_{K \backslash L}^{-1} \circ A)}\det \mathscr{R}_c(K_f,L_2).
\end{split}
\end{equation}
\section{From identities for discrete Laplacians to BFK-identities}\label{approximation}

This section is devoted to a proof of Theorem \ref{main2}.
Here we basically apply the techniques
by Dodziuk \cite{Dod:FDA} and M\"uller \cite{Mue:ATA}.

We recall the setup and notation laid out in \S \ref{relating-gluing}.
Let $(M,g)$ be a compact Riemannian manifold
with boundary $\partial M$ that consists of two disjoint boundary components
$N_1,N_2$.\footnote{By locality of the
argument, we may assume without loss of generality that $N_3, L_3 = \varnothing$.
} Consider a flat Hermitian vector bundle $(E,h, \A)$. Flatness implies product structure over a collar neighborhood of $\partial M$. We denote by $\Delta_{M}$ the Laplace Beltrami operator acting of functions on $M$ with values in $E$ and with Dirichlet boundary conditions at the boundary.

Assume that $g$ is product near $N_1$ and $N_2$ and define the closed double $\widetilde{M}$ by gluing a second copy of $M$ along the boundary. Let $N_1$ and $N_2$
be identified by an isometry $f$, and denote by $M_f$ the closed Riemannian manifold obtained
from $M$ by gluing $N_1$ onto $N_2$. The flat Hermitian vector bundle $(E,h, \A)$
induces smooth flat vector bundles over $M_f$ and $\widetilde{M}$.
We write $\Delta_{\widetilde{M}}$ and $\Delta_{M_f}$ for the twisted Laplacians on $\widetilde{M}$ and $M_f$, respectively, with Dirichlet boundary conditions at the respective boundaries.

Consider a simplicial complex $K$ which triangulates $M$ with
subcomplexes $L_1, L_2$ triangulating $N_1, N_2$, respectively.
Its double $\widetilde{K}$ along $L_1 \sqcup L_2$ is a simplicial decomposition of $\widetilde{M}$.
The simplicial complex $K_f$, obtained by gluing $K$ along the two identified
boundary components $L_1$ and $L_2$, decomposes $M_f$.

The pullbacks of the combinatorial analog of $E$ over $K$ define combinatorial vector bundles
over $K_f$ and $\widetilde{K}$. The combinatorial Riemannian structure on $K$ defines natural
combinatorial Riemannian structures on $K_f$ and $\widetilde{K}$. The metric structure on $M$
and the Whitney map define a combinatorial Riemannian structure on $K$ together with
the inner product on corresponding cochains with
values in combinatorial vector bundle.
Denote by $\Delta_{K_f}$ and $\Delta_{\widetilde{K}}$ the combinatorial
Laplace operators on cochains of degree zero, defined with respect to
this inner product, with Dirichlet boundary conditions at the
boundary.

Consider any covering $\{U_\A\}_{\A\in A}$ of $M_f$ by open subsets, such that the closure $\overline{U}_\A$
of each is a submanifold of $M_f$ with smooth boundary $\partial \overline{U}_\A$. Choose a subordinate partition
of unity $\{\phi_\A\}$. Let $\{\psi_\A\}$ be a family of functions $\psi_\A \in C^\infty_0(U_\A)$
with compact support in $U_\A$ such that $\psi_\A \restriction \textup{supp}\phi_\A \equiv 1$.

Let $\Delta_\A$ denote closure of the Hodge Laplace operator on $C^\infty(U_\A, E)$
with absolute boundary conditions\footnote{The arguments of this section hold similarly for Hodge Laplacians on differential forms.}.
Denote by $H_\A: L^2(U_\A, E) \to \ker \Delta_\A$
the corresponding harmonic projection and set $D_\A:= \Delta_\A + H_\A$. Then
\cite[Definition 8.10]{Mue:ATA} defines a pseudo-differential operator
\begin{align}\label{E}
E_{M_f}(s) := \sum_{\A\in A} \phi_\A D_\A^{-s} \psi_\A,
\end{align}
such that the difference between $(\Delta_{M_f} + H)^{-s}$ and $E_{M_f}(s)$ is smoothing
for every $s\in \C$, cf. \cite[Theorem 8.11]{Mue:ATA}.
Here, obviously $\Delta_{M_f}$ is the Hodge Laplacian on $\Omega^0(M_f,E)$
and $H$ the corresponding harmonic projection.

Given a smooth triangulation $K_f$ of $M_f$, we may choose an \emph{admissible} covering $\{U_\A\}_{\A\in A}$ of $M_f$,
such that $K_f$ induces a smooth triangulation $K_\A$ of each submanifold $\overline{U}_\A \subset M$
and its boundary $\partial \overline{U}_\A$. Let $\Delta^c_\A$ be the discrete Laplacian on
$C^0(K_\A, E)$, $H^c_\A$ the corresponding harmonic projection. Write $D^c_\A:= \Delta^c_\A + H^c_\A$.
We denote by $\Delta_{K_f}$ the discrete Laplacian on $C^0(K_f,E)$.
Recall the de Rham map from \eqref{de-Rham}. Then \cite[Definition 8.12]{Mue:ATA} defines
a combinatorial parametrix
\begin{align}
E^c_{K_f}(s) := \sum_{\A\in A} \phi_\A W (D^c_\A)^{-s} A \psi_\A.
\end{align}
Its trace is defined as follows. Let $\{a_l\}_{l=1}^N$ be an orthonormal basis
of $C^0(K_f, E)$ with the scalar product induced by the Whitney map $W$
respective to the given triangulation $K_f$. Then
\begin{align}
\textup{Tr}E^c_{K_f}(s) := \sum_{l=1}^N\langle E^c_{K_f}(s) W a_l, W a_l \rangle_{L^2(M_f, E)}.
\end{align}
The crucial property of the presented construction is the following
\begin{thm}\label{conv}\cite[Theorem 8.43, 8.44]{Mue:ATA}
The family $\zeta(s, \Delta_{K_f}) - \textup{Tr}E^c_{K_f}(s)$ is holomorphic in $s\in \C$
and locally uniformly (in $K_f$) bounded. Moreover, as the mesh $\delta>0$ of the
triangulation $K_f$ goes to zero under standard subdivisions, $\zeta(s, \Delta_{K_f}) - \textup{Tr}E^c_{K_f}(s)$
converge uniformly on every compact subset of $\C$ to the holomorphic
function $\zeta(s, \Delta_{M_f}) - \textup{Tr}E_{M_f}(s)$.
\end{thm}

Consider an admissible covering $\{U_\A\}_{\A\in A}$ of $M_f$, with $A=A_0 \dot\cup A_1$
such that $\{U_\A\}_{\A\in A_0}$ covers a collar of the hypersurface $N_1 \cong N_2$, 
and $U_\A \subset M \backslash \partial M$ for each $\A\in A_1$.
Then by locality of the individual summands in
\eqref{E} we find
\begin{equation}\label{123}
\begin{split}
E_{M_f}(s) = \sum_{\A\in A} \phi_\A D_\A^{-s} \psi_\A = 
\sum_{\A \in A_0} \phi_\A D_\A^{-s} \psi_\A + 
\sum_{\A \in A_1} \phi_\A D_\A^{-s} \psi_\A
= \frac{1}{2} E_{\widetilde{M}}(s) .
\end{split}
\end{equation}

In order to establish a similar relation on the combinatorial level, write
$P: L^2(M_f,E) \to WC^0(K_f,E)$ and $P_\A: L^2(U_\A,E) \to WC^0(K_\A,E)$
for the global and local orthogonal projections onto the image of the Whitney map, respectively.
Then by definition $\textup{Tr}E^c_{K_f}(s)= \textup{Tr}E^c_{K_f}(s)P$ and moreover,
the computations in the third displayed equation of \cite[p. 296]{Mue:ATA}
assert that for each $\A \in A$
\begin{align}
|\textup{Tr}(\phi_\A W (D^c_\A)^{-s} A \psi_\A P) - \textup{Tr}(\phi_\A W (D^c_\A)^{-s} A \psi_\A P_\A)| =\epsilon(\delta, s),
\end{align}
converges uniformly to zero on compact subsets of $s\in \C$ as the mesh $\delta>0$ of the
triangulation $K_f$ goes to zero under standard subdivisions.
The terms $\phi_\A W (D^c_\A)^{-s} A \psi_\A P_\A$ are local and the argument
of \eqref{123} applies. This proves
\begin{align}\label{123-c}
E_{K_f}(s) = \frac{1}{2} E_{\widetilde{K}}(s)  + \epsilon'(\delta, s),
\end{align}
where as above $\epsilon'(\delta, s)$ converges uniformly to zero on compact subsets of $s\in \C$ as the mesh $\delta>0$ of the
triangulation $K_f$ goes to zero under standard subdivisions. Combining Theorem \ref{conv} with
the relations \eqref{123} and \eqref{123-c} proves

\begin{thm}\label{main2-b}
As the mesh $\delta>0$ of the triangulation $K$ goes to zero
under standard subdivisions
\begin{equation}
 \begin{split}
\lim_{\delta \to 0}  \frac{(\det' \Delta_{K_f})^2}
{\det' \Delta_{\widetilde{K}}}
= \frac{(\det\nolimits_\zeta \Delta_{M_f})^2}
{\det\nolimits_\zeta \Delta_{\widetilde{M}} }
   \end{split}.
\end{equation}
\end{thm}

As an obvious consequence of Theorem \ref{main2-b}, \eqref{gluing-loc-formula} and Theorem \ref{BFK-theorem} we arrive at the following 
relation between the analytic and combinatorial DN operators

\begin{align*}
&\lim_{\delta \to 0} \, \frac{(\det'\mathscr{R}_c(K_f,L_2))^2}
{\det'\mathscr{R}_c(\widetilde{K},L_1 \sqcup L_2)} \cdot 
\frac{\det \left(Q_{L_2\sqcup L_3} - \widetilde{A}^t \circ Q^{-1}_{(K\backslash L) \sqcup (K\backslash L)} \circ \widetilde{A}\, \right)}
{\det (Q_{L_2} -A^t \circ Q^{-1}_{K\backslash L} \circ A)^2}
\\ &= \frac{\mathscr{C}^2_{M_f,N_2}}{\mathscr{C}_{\widetilde{M},N_1\sqcup N_2}}
\cdot \frac{ (\det\nolimits_\zeta \mathscr{R}_a(M_f,N_2))^2}
{\det\nolimits_\zeta \mathscr{R}_a(\widetilde{M},N_1\sqcup N_2)},
\end{align*}
where $\widetilde{A}=A\oplus A$. Corollary \ref{main3} now follows from the next 

\begin{prop} For $\delta>0$ sufficiently small
$$\frac{\det \left(Q_{L_2\sqcup L_3} - \widetilde{A}^t \circ Q^{-1}_{(K\backslash L) \sqcup (K\backslash L)} \circ \widetilde{A}\, \right)}
{\det (Q_{L_2} -A^t \circ Q^{-1}_{K\backslash L} \circ A)^2}=1.$$
\end{prop}

\begin{proof}
Denote by $\mathscr{U}_L \subset K_f$ the metrized collar complex of $L_2\subset K_f$. Let $v\in C^0(K\backslash L)$ be an element in the image of $A$. By locality of 
the Whitney inner product, $v$ is supported inside the star neighborhood $St(L_2)$ of $L_2\subset K_f$, which consists of all those 
vertices $\tau \in K_f$, which are connected to $L_2$ by an edge. For any $\omega \in C^0(K\backslash L)$
$$
\langle \omega, Q^{-1}_{K\backslash L} v\rangle_K = \langle \omega, v \rangle_0
$$
and equals zero, if and only if $\omega$ is supported in $K_f \backslash St(L_2)$. By locality of the Whitney inner product,
this in turn implies that $Q^{-1}_{K\backslash L} v$ is supported inside $St(St(L_2))$. 
Assume $\delta>0$ is sufficiently large, so that $Q_{K_f}(St(St(L_2)))\subset \mathscr{U}_L$. Then we find
$$
Q_{\mathscr{U}_L\backslash L_2} [Q^{-1}_{K\backslash L} v] = Q_{K\backslash L} [Q^{-1}_{K\backslash L} v] = v,
$$
and hence $Q^{-1}_{K\backslash L} v = Q^{-1}_{\mathscr{U}_L\backslash L_2} v$. Consequently
$Q^{-1}_{K\backslash L} \circ A$ depends only on the metrization over $\mathscr{U}_L$ and the statement follows immediately
from the product metric structure assumption. 
\end{proof}

\section{Conclusion}\label{research}

This paper is a step towards constructing the quantum field theory of free scalar Bose field on a
Riemannian manifold $M$ as a limit of Gaussian quantum field theories on finite metrized simplicial
approximations of $M$ when the mesh of the approximation goes to zero.
In quantum field theory and statistical mechanics such limit is known as a scaling limit
(near a point of phase transition, which is mesh equals to zero in our case). The main step
towards completion of such program is the characterization of the zero mesh limit of determinants
of combinatorial Laplacians. One should expect that as mesh $\epsilon$ goes to zero
\[
\log\det(\Delta_{K_\epsilon}+m^2)=2 N_\epsilon + \sum_{j=1}^{n-1} c_j\epsilon^j +
\mbox{ log. terms } +c_0+\log\det\nolimits_\zeta(\Delta+m^2)+ o(1).
\]
Here $N_\epsilon$ is the number of vertices in $K_\epsilon$, $n$ is the dimension of the simplicial complex,
constants $c_0,\dots, c_{n-1}$ are not universal, i.e.
they depend on $K_\epsilon$. This problem is largely open. For some results in this direction see
\cite{Ken:TAD}, \cite{CJK:ZFH}, \cite{AS}.

Another problem, closely related to this paper is the construction of
of topological quantum field theories based on an approximation of
space times by simplicial complexes. An example of such TQFT was constructed in \cite{Mn}.
In this case we will have combinatorial torsions instead of determinants of
Laplacians. Discrete version of the De Rham differential and of the exterior
product are given by corresponding $A_\infty$ algebras. A very important but
entirely understood question in this direction, see e.g. Wilson \cite{Wilson}, is what is the most natural
discrete counterpart of Riemannian geometry (in particular of Hodge star operation).

We did not discuss here first order formulation of the classical
field theory on simplicial complexes (see for example \cite{CMR} for
classical field theories on manifolds). We will do it in a separate
publication.

The approximation of space time by a complex for special
box complexes is well known in constructive field theory.
In this sense the present paper can be regarded as a development in
constructive quantum field theory with the aim to construct
an Atiyah-Segal style quantum field theory. In dimension $2$
this was done in \cite{P:CFT} for quantum $P(\phi)_2$ theory directly
in the continuum case. It would be interesting to derive these results
from simplicial approximations.

\providecommand{\bysame}{\leavevmode\hbox to3em{\hrulefill}\thinspace}
\providecommand{\MR}{\relax\ifhmode\unskip\space\fi MR }
\providecommand{\MRhref}[2]{%
  \href{http://www.ams.org/mathscinet-getitem?mr=#1}{#2}
}
\providecommand{\href}[2]{#2}


\begin{thebibliography}{BFK92}

\bibitem[Ati90]{At} Michael Atiyah, \textit{On framings of $3$-manifolds}, Topology 29 1 (1990) 1--7.

\bibitem[BrLe92]{BruLes:HK}
Br\"uning, J.; Lesch, M. \emph{Hilbert complexes}, J. Funct. Anal. 108 (1992),
no. \textbf{1}, 88--132. \MR{1174159 (93k:58208)}

\bibitem[BFK92]{BFK:MVT}
D.~Burghelea, L.~Friedlander, and T.~Kappeler, \emph{Meyer-{V}ietoris type
  formula for determinants of elliptic differential operators}, J. Funct. Anal.
  \textbf{107} (1992), no.~1, 34--65. \MR{1165865 (93f:58242)}

\bibitem[CMR11]{CMR}  A. S. Cattaneo, P. Mnev, N. Reshetikhin, \emph{Classical and quantum Lagrangian field theories with boundary}, PoS(CORFU2011)044,
 arXiv:1207.0239 [math-ph].

\bibitem[CJK10]{CJK:ZFH}
Gautam Chinta, Jay Jorgenson, and Anders Karlsson, \emph{Zeta functions, heat
  kernels, and spectral asymptotics on degenerating families of discrete tori},
  Nagoya Math. J. \textbf{198} (2010), 121--172. \MR{2666579 (2011i:58052)}

\bibitem[ClMa04]{ClaMaz:CSA}
Jonathan Claridge and Rafe Mazzeo, \emph{Connected sums and generic properties
  in spectral graph theory}, bachelor thesis, Stanford (private communication)

\bibitem[Dod76]{Dod:FDA}
Jozef Dodziuk, \emph{Finite-difference approach to the {H}odge theory of
  harmonic forms}, Amer. J. Math. \textbf{98} (1976), no.~1, 79--104.
  \MR{0407872 (53 \#11642)}
	
\bibitem[DoPa76]{DodPat}	
Dodziuk, J.; Patodi, V. K. 
\emph{Riemannian structures and triangulations of manifolds}, J. Indian Math. Soc. (N.S.) 
\textbf{40} (1976), no. 1-4, 1--52 (1977).

\bibitem[For92]{For:FDA}
Robin Forman, \emph{Functional determinants and geometry},
  Invent. Math.\ {\bf 88} (1987), no. 3, 447--493; \MR{0884797
  (89b:58212)}
  
\bibitem[Ken00]{Ken:TAD}
Richard Kenyon, \emph{The asymptotic determinant of the discrete {L}aplacian},
  Acta Math. \textbf{185} (2000), no.~2, 239--286. \MR{1819995 (2002g:82019)}

\bibitem[Lee03]{Lee:BFK}
Yoonweon Lee, \emph{Burghelea-{F}riedlander-{K}appeler's gluing formula for the
  zeta-determinant and its applications to the adiabatic decompositions of the
  zeta-determinant and the analytic torsion}, Trans. Amer. Math. Soc.
  \textbf{355} (2003), no.~10, 4093--4110 (electronic). \MR{1990576
  (2004e:58058)}

\bibitem[Les97]{Les:OFT}
Matthias Lesch, \emph{Operators of {F}uchs type, conical singularities, and
  asymptotic methods}, Teubner-Texte zur Mathematik [Teubner Texts in
  Mathematics], vol. 136, B. G. Teubner Verlagsgesellschaft mbH, Stuttgart,
  1997. \MR{1449639 (98d:58174)}

\bibitem[Mne09]{Mn} Pavel Mnev \emph{Discrete BF theory}, preprint, arXiv:0809.1160.

\bibitem[M{\"u}l78]{Mue:ATA}
Werner M{\"u}ller, \emph{Analytic torsion and {$R$}-torsion of {R}iemannian
  manifolds}, Adv. in Math. \textbf{28} (1978), no.~3, 233--305. \MR{498252
  (80j:58065b)}

\bibitem[Seg04]{Se}  Graeme Segal, \emph{The Definition of Conformal Field Theory}, 
In: Topology, Geometry and Quantum Field Theory,  London Mathematical Society Lecture 
Note Series (No. 308), Cambridge University Press, 2004, pp. 421--577.

\bibitem[Sri14]{AS} Ananth Sridhar, \emph{Asyptotic of determinant of discrete Laplacians}, preprint.

\bibitem[PaWo05]{ParWoj:ADO}
Jinsung Park and Krzysztof~P. Wojciechowski, \emph{Adiabatic decomposition of
  the {$\zeta$}-determinant and {D}irichlet to {N}eumann operator}, J. Geom.
  Phys. \textbf{55} (2005), no.~3, 241--266. \MR{2160038 (2006b:58033)}

\bibitem[Pic08]{P:CFT}
Dough Pickrell, \emph{$P(\phi)_2$ Quantum Field Theories and Segal's Axioms},
Comm. Math. Phys. \textbf{280} (2008), 403--425.

\bibitem[Whi57]{Whi:GIT}
Hassler Whitney, \emph{Geometric integration theory}, Princeton University
  Press, Princeton, N. J., 1957. \MR{0087148 (19,309c)}

\bibitem[Wil07]{Wilson}
Wilson, Scott O. \emph{Cochain algebra on manifolds and convergence under refinement}, 
Topology Appl. \textbf{154} (2007), no. 9, 1898--1920.

\end{thebibliography}
\end{document}